\documentclass[a4paper,11pt]{article}
\usepackage{amsmath,amsthm,mathrsfs}
\usepackage{dsfont}
 \usepackage{etex}
\usepackage{bigints}
\usepackage[utf8]{inputenc}
\usepackage{mathtools,amsmath,amssymb}
\usepackage{lmodern}
\usepackage[T1]{fontenc}
\usepackage{microtype}
\usepackage[pagebackref=false]{hyperref}
\numberwithin{equation}{section}

\usepackage[dvipsnames]{xcolor}

\usepackage{xspace}
\usepackage{bibentry}

\usepackage{graphicx}

\hypersetup{pdfborder={0 0 0},colorlinks,breaklinks=true,
  urlcolor={TealBlue},citecolor={Magenta},linkcolor={RoyalBlue}
}

\usepackage[low-sup]{subdepth}

\usepackage{subcaption}
\usepackage{tikz}
\usetikzlibrary{decorations.pathreplacing}
\usepackage[margin=2.5cm]{geometry}

\newcommand{\chiara}[1]{{(* {\color{RedViolet}Ch:{} \small #1}*)}}
\newcommand{\rouven}[1]{{ (* {\color{BurntOrange}R:{} \small #1}*)}}
  
\def\be{\begin{equation}}
\def\ee{\end{equation}}
\def\bea{\begin{eqnarray}}
\def\eea{\end{eqnarray}}

\def\<{\langle}
\def\>{\rangle}
\def\~{\tilde}

\newcommand{\R}{\mathbb R}

\newcommand{\N}{\mathbb N}

\newtheorem{theorem}{Theorem}[section]

\newtheorem{definition}[theorem]{Definition}

\newtheorem*{remark}{Remark}

\usetikzlibrary{decorations.markings}

\begin{document}

\begingroup\parindent0pt
\centering
\begingroup\LARGE
\bf
Integrable heat conduction model
\par\endgroup
 \vspace{3.5em}
 \begingroup\large
 {\bf Chiara Franceschini}, 
 {\bf Rouven Frassek}, 
 {\bf Cristian Giardin\`{a}}
 \par\endgroup
\vspace{2em}
%
University of Modena and Reggio Emilia,\\
Department of Mathematics,\\
 Via G. Campi 213/b, 41125 Modena, Italy\\
\vspace{5em}


 \vspace{1.cm}
 \begin{abstract}
 \noindent
We consider a stochastic process of heat conduction where energy is redistributed along a chain between nearest neighbor sites via an improper beta distribution. Similar to the well-known Kipnis-Marchioro-Presutti (KMP) model, the finite chain is coupled at its ends with two reservoirs that break the conservation of energy when working at different temperatures. At variance with KMP, the model considered here is integrable and one can  write in a closed form the $n$-point  correlation functions of the non-equilibrium steady state. As a consequence of the exact solution one can directly prove that the system is in a `local equilibrium' and  described at the macro-scale by a product measure. 
Integrability manifests itself through the  description of the model via the open Heisenberg chain with non-compact spins. The algebraic formulation of the model allows to interpret its duality relation with a purely absorbing particle system as a change of representation.
 \end{abstract}

\endgroup

\thispagestyle{empty}
\setcounter{tocdepth}{2}


\newpage
\section{Introduction}

In their paper from 1982 \cite{kipnis1982heat}, 
Kipnis, Marchioro and Presutti (KMP) introduced
a solvable model of heat conduction. They considered a chain of
energy-exchanging harmonic oscillators, that is coupled 
 to Gibbs reservoirs 
of different temperatures at the two ends of the chain.
They successfully implemented the idea
of modelling the non-harmonic effects
by a stochastic process which redistributes 
`microcanonically' the energy between nearest-neighbor oscillators, i.e. keeping constant their total energy. The pioneering idea behind
the KMP model was substantially expanded in several
later works and found applications in different
contexts such as fluctuations and large deviations
of the temperature profile and of the current
\cite{krapivsky2012fluctuations,zarfaty2016statistics,bettelheim2022inverse, bertini2015macroscopic},
dynamical phase transitions \cite{espigares2013dynamical},
anomalous heat conduction of 
momentum conserving systems
\cite{bernardin2005fourier,giardina2005fourier, 
basile2006momentum, basile2009thermal}, 
duality theory of Markov processes and hidden symmetries
\cite{giardina2009duality,carinci2013duality}.

The KMP model was thus the first mathematical
model of energy-transport where it was possible to rigorously 
prove that the non-equilibrium steady state has the local equilibrium property and furthermore Fourier's law with a constant conductivity holds. The main reason behind 
the solvability of the KMP
process is ``duality'', i.e. the existence of an
associated dual process that
is simpler to study. By duality, the computation
of the $n$-point correlations functions in the
non-equilibrium steady state is reduced to the 
problem of the exit distribution of $n$ dual particles.
Propagation of chaos and local equilibrium were then
obtained by showing that, in the macroscopic limit,
the dual particles essentially behave like
independent particles \cite{kipnis1982heat}.

However, the full microscopic
description of the non-equilibrium stationary state
remained inaccessible in the KMP model,
as the absorption probabilities of the
dual particles have not been computed in closed-form.
In this paper we consider an energy-redistribution
model which, while being similar to the one of KMP,
allows a full description {\em at microscopic level}
of the non-equilibrium steady state.
More precisely, for the integrable heat conduction model considered in this paper we {compute all the multivariate moments of the stationary measure (non-equilibrium steady state).} From the exact solution we establish that
when the reservoir parameters are equal (equilibrium) we have a Gibbs measure corresponding to a product of Gamma distributions.
If instead the reservoir parameters are different (non-equilibrium) we find
long-range correlations.
Local equilibrium, i.e.
a product measure around
each {\em macroscopic} point with inhomogeneous parameters interpolating between the reservoir parameters, is then recovered as a consequence of the exact solution.

The closed form expression for the
multivariate moments paves the way
to the rigorous study of density fluctuations
(both typical fluctuations as well as
large deviations) for non-equilibrium steady states \cite{LD},
with the possibility of a direct comparison  to the predictions of
the Macroscopic Fluctuation Theory (MFT) \cite{bertini2015macroscopic} developed
for boundary driven diffusive systems.

\paragraph{Added note.} In the progress of this work, we
became aware of the interesting article
\cite{CGT} that relates
to the model studied in this paper. More precisely, the authors 
considered the most general 
model of energy redistribution with
a redistribution rule that is of ``zero-range''
type, namely the rate at which the energy
is moved between two sites is just a function
of the energy at the site where it is taken off.
Remarkably, within this large class of models,
the authors show that the model with generator \eqref{eq:gen-bulk} studied 
in this paper, as well as the model with generator \eqref{eq:gen-bulk2} studied in \cite{Frassek:2019vjt,frassek2021exact},
emerge as the only two models of ``zero-range'' type where the following three properties are fullfilled:
\textit{i)} they are of gradient type; \textit{ii)} they have
a product measure at equilibrium; 
\textit{iii)} they have constant diffusivity and
quadratic mobility, which in turn guarantees
the possibility of computing
large deviation functions using the approach of
 Macroscopic Fluctuation Theory \cite{bertini2015macroscopic}. Thus, within the class of ``zero-range'' energy redistribution models, the explicit  writing of the microscopic correlation functions in the non-equilibrium steady state (a fact that is rooted in the duality and integrability of the models) seems to be related to the solvability of the variational problem emerging in the study of large deviation functions via MFT.

\paragraph{Paper organization.}
The  paper is organized as follows.
In Section~\ref{sec2.1} we define our model and state our main results, first in the equilibrium set-up (Section~\ref{sec2.3}), then for the non-equilibrium setting (Section~\ref{sec2.4}). We also comment on the relation of our model to 
other energy and/or particle redistribution models previously considered in the literature in Section~\ref{other-models}. The proof of our results are contained in Section~\ref{sec:proofs},
which contains the verification of duality 
via a direct computation
and the proof of the main results.
Finally in Section~\ref{sec4} duality is derived from the algebraic structure
of the model, which is essentially given by the open XXX Heisenberg chain with non-compact spins.
In  Appendix~\ref{appendixa} a direct
proof of reversibility at equilibrium is given, while  Appendix~\ref{B} contains some formulas used in the proof of the
 algebraic structure of the model.

\section{The model and the main result}
\label{sec2}

\subsection{Model definition}
\label{sec2.1}
We consider a one-dimensional lattice system of interacting particles. For $i\in\{1,\ldots, N\}$
we denote by  $y_i\in \R_+:=(0,\infty)$ the energy
of the $i^{th}$ particle
and we collect all the energies into the vector $y=(y_1,\ldots,y_N)\in \R_+^N$.

The system undergoes a stochastic evolution  defined by a Markov process $\{y(t)\,,\,t\ge 0\}$
with generator working on function $f:\R_+^N\to \R$ given by
\be
\label{eq:gen}
\mathscr{L} = {\mathscr{L}_1} + \sum_{i=1}^{N-1} {\mathscr{L}}_{i,i+1} + {\mathscr{L}}_N\;.
\ee
The bulk term ${\mathscr{L}}_{i,i+1}$,
which depends on a parameter $s > 0$, describes the exchange of energy between the particles on the bond $(i,i+1)$; it will be convenient to further split it as
\begin{equation}\label{eq:gen-bulk}
{\mathscr{L}_{i,i+1}} = {\mathscr{L}^{\rightarrow}_{i,i+1}} + {\mathscr{L}^{\leftarrow}_{i,i+1}}
\end{equation} 
 where
\begin{align} 
\label{eq:gen-bulk-right}
{\mathscr{L}^{\rightarrow}_{i,i+1}}f(y)  &=
\int_0^{y_i}\frac{d\alpha}{\alpha} \left(1-\frac{\alpha}{y_i}\right)^{2s-1}\left[f(y_1,\ldots,y_i-\alpha,y_{i+1}+\alpha,\ldots,y_N)-f(y)\right]\,, \\
\intertext{and} 
\label{eq:gen-bulk-left}
{\mathscr{L}^{\leftarrow}_{i,i+1}}f(y) &=
\int_0^{y_{i+1}}\frac{d\alpha}{\alpha} \left(1-\frac{\alpha}{y_{i+1}}\right)^{2s-1}\left[f(y_1,\ldots,y_i+\alpha,y_{i+1}-\alpha,\ldots, y_N)-f(y)\right]\;.
\end{align}
Thus $\mathscr{L}^{\rightarrow}_{i,i+1}$ describes
the stochastic movement of energy towards the right, i.e. from site $i$ to site $i+1$, and similarly $\mathscr{L}^{\leftarrow}_{i,i+1}$ describes
the movement of energy towards the left.
Obviously, the process generated by the bulk generator preserves the total energy 
$\sum_{i=1}^N y_i$.
For a discussion about the origin and interpretation of
such generator see Section~\ref{other-models}.

The boundary terms, which depend on additional parameters 
$\lambda_L>0,\lambda_R >0$, are given by 
\begin{equation}  \label{eq:leftbnd}
\begin{split}
\mathscr{L}_{1}f(y)&
=
\int_0^{y_1}\frac{d\alpha}{\alpha} \left(1-\frac{\alpha}{y_1}\right)^{2s-1}\left[f(y_1-\alpha,\ldots,y_N)-f(y)\right] \\
&\quad + 
 \int_0^{+\infty}\frac{d\alpha}{\alpha} e^{- \lambda_L \alpha} \left[f(y_1+\alpha,\ldots,y_N)-f(y)\right]\,,
 \end{split}
 \end{equation}  
and 
 \begin{equation} \label{eq:rightbnd}
 \begin{split}
\mathscr{L}_{N}f(y)
&=
\int_0^{y_N}\frac{d\alpha}{\alpha} \left(1-\frac{\alpha}{y_N}\right)^{2s-1}\left[f(y_1,\ldots,y_N-\alpha)-f(y)\right] \\
& \quad+ 
 \int_0^{+\infty}\frac{d\alpha}{\alpha} e^{- \lambda_R \alpha} \left[f(y_1,\ldots,y_N+\alpha)-f(y)\right]\,.
 \end{split}
 \end{equation}  
They model the reservoirs which insert and remove energy at the left and right ends of the chain. 

Thus, the full process generated by
\eqref{eq:gen} takes value on $\R_+^N$
and it does not conserve the total
energy. In fact, it is a model
for heat conduction (or energy transport) 
from one side of the chain to the other.

\begin{remark}
 In the following we will 
 consider the action of the 
 generator \eqref{eq:gen}
on polynomial functions.
We do not address here the
problem of characterizing
the domain of the generator.
\end{remark}

\subsection{Comparison to other energy-redistribution models}
\label{other-models}

The stochastic energy redistribution rule encoded in the bulk generator \eqref{eq:gen-bulk} can be described as follows. Fix a couple $(i,i+1)$ and consider the movement of energy to the right described by 
${\mathscr{L}^{\rightarrow}_{i,i+1}}$ (a similar reasoning
can be made for the movement of energy to the left described by ${\mathscr{L}^{\leftarrow}_{i,i+1}}$).
For $a,b>0$ let $U$ be  a random variable with 
$\text{Beta}(a,b)$ distribution, i.e. having probability density
$$
\sigma_{a,b}(u) = \frac{u^{a-1} \left(1-u\right)^{b-1}}{B(a,b)}\,,
$$ 
where $B(a,b) = \frac{\Gamma(a)\Gamma(b)}{\Gamma(a+b)}$
is the Beta function.
Consider the process where the energy at site $i$
evolves as follows: 
at rate 1 energy is exchanged
 between the pair $(i, i + 1)$
 by drawing a number $0 \le u \le 1$, independently of everything else, according to the $\text{Beta}(a,b)$ distribution and
a fraction $u y_i$ of  the energy at site $i$ is moved to site $i + 1$ and all other energies remain unchanged.
This is described by the generator $L^{\rightarrow, a,b}_{i,i+1}$ defined as
\begin{align} 
\label{genab}
L^{\rightarrow, a,b}_{i,i+1}f(y)  &=
\int_0^{1}\sigma_{a,b}(u)\left[f(y_1,\ldots,y_i-y_i u,y_{i+1}+y_i u,\ldots,y_N)-f(y)\right] du \nonumber\\
&  = \mathbb{E}\left[f(y_1,\ldots,y_i- y_i U,y_{i+1}+ y_i U,\ldots,y_N)-f(y)\right]\,.
\end{align}
Speeding up the time by a factor $B(a,b)$ and applying the change
of variable $u= \alpha/y_i$, then one can see that
\be
{\mathscr{L}^{\rightarrow}_{i,i+1}}  = \lim_{\substack{a\to 0\\[0.05cm]\,\,\,b\to 2s}} B(a,b) \cdot L^{\rightarrow, a,b}_{i,i+1}\,.
\ee
Thus one gets that ${\mathscr{L}^{\rightarrow}_{i,i+1}}$ corresponds to a process with infinite intensity, which is indeed a pure jump Levy process with Levy measure $\nu(d\alpha)= \frac{1}{\alpha}\left(1-\frac{\alpha}{y_i}\right)^{2s-1} d\alpha$.

\bigskip
\noindent
The description above of the process allows its comparison to other models of energy redistribution.

\paragraph{KMP model.}
The original KMP model \cite{kipnis1982heat} was defined by a uniform
redistribution rule for the sum of the energies
of nearest neighbour  sites and by instantaneous thermalization with reservoirs imposing an exponential (Gibbs) distribution at the boundaries. Here we consider its
generalized version, introduced in \cite{carinci2013duality},
where the redistribution
rule is given in terms of a $\text{Beta}(2s,2s)$
random variable and the 
reservoirs thermalize with Gamma distributions of shape parameter $2s>0$, so that the original KMP 
is recovered for $s=1/2$. Namely, we consider the 
process with generator
\be
\label{eq:gen-kmp}
\mathscr{L}^{\text{KMP}} = {\mathscr{L}_1}^{\text{KMP}} + \sum_{i=1}^{N-1} {\mathscr{L}}_{i,i+1}^{\text{KMP}} + {\mathscr{L}}_N^{\text{KMP}}
\ee
where we have
$$ 
{\mathscr{L}}_{i,i+1}^{\text{KMP}}f(y)  =
\int_0^{1}
\sigma_{2s,2s}(w)\left[f(y_1,\ldots, w(y_i+y_{i+1}), (1-w)(y_i+y_{i+1}),\ldots,y_N)-f(y)\right] dw\,,
$$
while 
\be 
\label{kmpboundary1}
\mathscr{L}^{\text{KMP}}_{1}f(y)
=
 \int_0^{+\infty}{dy_1'}\frac{({\lambda_L})^{2s}}{\Gamma(2s)} (y_1')^{2s-1}e^{- \lambda_L y_1'} \left[f(y_1',\ldots,y_N)-f(y)\right]
\ee  
and 
\be
\label{kmpboundary2}
\mathscr{L}^{\text{KMP}}_{N}f(y)
=
 \int_0^{+\infty}dy_N'\frac{({\lambda_R})^{2s}}{\Gamma(2s)} (y_N')^{2s-1} e^{- \lambda_R y_N'} \left[f(y_1,\ldots,y_N')-f(y)\right]\,.
\ee
Comparing the bulk generators of our process to the
one of the KMP process we see that the generator
\eqref{eq:gen} is associated to a redistribution rule
of the type
\begin{equation*}
(y_i,y_{i+1}) \longrightarrow \left\{
\begin{array}{ll}
\bigg(y_i - y_i U\, ,\, y_{i+1} +  y_i U\bigg) & \text{with probability } \frac12 ,\\
&\\
\bigg(y_i + y_{i+1} V\, ,\,  y_{i+1} - V y_{i+1}\bigg) & \text{with probability } \frac12 ,
\end{array} \right.
\end{equation*}
whereas the KMP process is associated to 
$$
(y_i,y_{i+1}) \longrightarrow \bigg((y_i+y_{i+1})W\, ,\,(y_i+y_{i+1})(1-W) \bigg).
$$
Here $U$ denotes an ``improper'' $\text{Beta}(0,2s)$
distribution, $V$ denotes an independent copy of $U$, and $W$ denotes a $\text{Beta}(2s,2s)$
distribution.

As for the reservoirs, we remark that the boundary generators of our integrable heat transport (\ref{eq:leftbnd}, \ref{eq:rightbnd}) as well as those of the generalized KMP model (\ref{kmpboundary1}, \ref{kmpboundary2}) are both reversible with respect to the Gamma distribution, although they are substantially different. A proof of this can be found in Appendix  \ref{appendixa}.

\paragraph{Immediate exchange model.} The heat conduction of this paper gets its integrable structure from the 
improper Beta random variables used in the 
energy distribution rule.
The process constructed using the generator \eqref{genab} with standard Beta random random variables is known as the ``Immediate Exchange Model'', and it has been studied in the economics literature as a model of wealth redistribution.
See \cite{redig2017generalized,van2016duality} and references therein for the study of its duality properties and \cite{ sasada2015spectral} for the study of its spectral gap problem.

\paragraph{``Harmonic'' model.}
Our model arises as the scaling limit of an interacting particle systems, the
``harmonic'' model introduced in 
\cite{Frassek:2019vjt} , so-called because it involves harmonic numbers. The harmonic model  is integrable and was solved in \cite{frassek2021exact}.
In fact, as we shall prove later,
our model is also integrable and its solution
relies on the fact  that the boundary-driven
harmonic model and our model have {\em the same
absorbing dual process}.

To see the scaling limit, we recall that the harmonic process is an interacting particle system defined by the generator
\be
\label{eq:gen-har}
{\mathscr{L}}^{\text{Har}} =
{\mathscr{L}}^{\text{Har}}_1 +
\sum_{i=1}^{N-1} {\mathscr{L}}^{\text{Har}}_{i,i+1} + {\mathscr{L}}^{\text{Har}}_N\,,
\ee
where the bulk term reads
\begin{equation}\label{eq:gen-bulk2}
{\mathscr{L}^{\text{Har}}_{i,i+1}} = {\mathscr{L}^{{\text{Har}},\rightarrow}_{i,i+1}} + {\mathscr{L}^{{\text{Har}}, \leftarrow}_{i,i+1}}
\end{equation} 
with
$$
({\mathscr{L}^{{\text{Har}},\rightarrow}_{i,i+1}} f)(\eta) = \sum_{k=1}^{\eta_i}\varphi_s(k,\eta_i) \Big[f(\eta-k \delta_i + k \delta_{i+1}) - f(\eta)\Big]  
$$
and
$$
({\mathscr{L}^{{\text{Har}},\leftarrow}_{i,i+1}} f)(\eta) = 
\sum_{k=1}^{\eta_{i+1}} \varphi_s(k,\eta_{i+1}) \Big[f(\eta +k\delta_i - k \delta_{i+1}) - f(\eta)\Big] \,.
$$
Here $\eta\in \N^{N}$ denotes a particle configuration,
with $\eta_i$ being the number of particles at site $i$,
and $\delta_i$ denotes the configuration with only particle placed at site $i$. The rates $\varphi_s: \N\times \N \to \R_+$ are defined by
\be
\label{eq:varphi1}
\varphi_s(k,n) :=  \frac{1}{k}\frac{\Gamma (n+1) \Gamma (n-k+2 s)}{ \Gamma (n-k+1) \Gamma (n+2 s)} 
\mathds{1}_{\{1,2,\ldots,n\}}(k)
\ee
where $\mathds{1}_{A}$ is the indicator function
of the set $A$.
For the boundary terms we have
\be 
(\mathscr{L}^{\text{Har}}_{1}f)(\eta)
=
\sum_{k=1}^{\eta_1} \varphi_s(k,\eta_1) \Big[f(\eta-k\delta_1) - f(\eta)\Big]
 +   
\sum_{k=1}^{\infty} \frac{\beta_L^k}{k} \Big[f(\eta+k\delta_1) - f(\eta)\Big]
\ee  
and 
\be
(\mathscr{L}^{\text{Harm}}_{N}f)(\eta)
=
\sum_{k=1}^{\eta_N} \varphi_s(k,\eta_N) \Big[f(\eta-k\delta_N) - f(\eta)\Big]
 +   
\sum_{k=1}^{\infty} \frac{\beta_R^k}{k} \Big[f(\eta+k\delta_i) - f(\eta)\Big]\;.
\ee
Denote by $\{\eta(t)\,, t\ge 0\}$ the boundary-driven ``harmonic'' process
with generator \eqref{eq:gen-har}. 
Introducing
the scaling paramater $\epsilon >0$, one can show
that the process $\{\epsilon\eta(t)\,, t\ge 0\}$
with $\beta_L,\beta_R$ such that $\epsilon^{-1}\beta_L = \lambda_L(1+o(1))$ and
$\epsilon^{-1}\beta_R = \lambda_R(1+o(1))$
converges (weakly) to the process $\{y(t)\,,\,t\ge 0\}$
with generator \eqref{eq:gen}, as $\epsilon \to 0$, see the Appendix A of \cite{Frassek:2019vjt} for details.
Both for the bulk and for the boundary, the scaling is obtained by using the asymptotics
\begin{equation}
\frac{\Gamma[z+\gamma_1]}{\Gamma[z+\gamma_2]} 
\simeq z^{\gamma_1-\gamma_2} \left(1+O\left(\frac{1}{z}\right)\right) \qquad \text{as } z \to \infty\,.
\end{equation}

\subsection{Equilibrium}
\label{sec2.3}

For all system sizes $N$ and all choices of the reservoirs parameters $\lambda_L,\lambda_R >0$, the Markov process $\{y(t)\,,\,t\ge 0\}$ with generator \eqref{eq:gen} has a unique stationary measure, supported on $\R_+^N$, that we shall denote $\mu_{N}$. We avoid writing the dependence on $\lambda_L,\lambda_R$ to not burden the notation. 

When the system is in an equilibrium set-up, i.e when $\lambda_L = \lambda_R =\lambda$, the heat conduction model has an invariant Gibbs product measure.

{Recall that an invariant distribution for a process is a probability distribution that remains unchanged as time progresses, namely if $y(0) \sim \mu$, then for all $t>0$ $y(t) \sim \mu$. An invariant measure is also reversible if the generator is self-adjoint in $L_2(\R_+^N,\mu)$. More precisely in Appendix \ref{appendixa} we prove Hermiticity when the generator acts on a set of polynomial functions.}

\begin{theorem}[Equilibrium reversible measure]
\label{prop:eq}
If  $\lambda_L = \lambda_R = \lambda$
then the product probability measure with marginal 
the Gamma  distribution with shape
parameter $2s>0$ and rate parameter $\lambda>0$, i.e. 
\be \label{eq:density}
\mu_N(dy) = \prod_{i=1}^N 
\frac{\lambda^{2s}}{\Gamma(2s)} y_i^{2s-1}e^{-\lambda y_i} dy_i, 
\ee
is stationary.
Furthermore it is also  reversible.
\end{theorem}
Stationarity will be deduced as
a consequence of duality in Section~\ref{sec:proofs}. 
Reversibility is proven in  Appendix~\ref{appendixa}
by a direct computation showing that the generator \eqref{eq:gen}
is self-adjoint in the Hilbert space
$L_2(\R_+^N,\mu_N)$.

\subsection{Non-equilibrium steady state}
\label{sec2.4}

In a non-equilibrium set-up, i.e. when 
$\lambda_L\neq \lambda_R$, reversibility is lost and the stationary measure $\mu_{N}$ is non-product.
Our main result is the {computation of all multivariate moment of the stationary measure  in non-equilibrium.   }

\begin{theorem}[Non-equilibrium steady state]
\label{theo:main}
Define the left and right ``reservoir temperatures'' as
\be
\label{eq:rho}
T_L = \frac{1}{\lambda_L} \qquad\text{and}\qquad T_R = \frac{1}{\lambda_R}.
\ee
Then we have:
\begin{itemize}
\item 
For a multi-index $\xi=(\xi_1,\ldots,\xi_N)\in\N_0^N$, the multivariate moments of the stationary measure 
of the process generated by \eqref{eq:gen} are given by
\begin{equation}
\label{eq:moments}
\bigint_{\R_+^N}\left[\prod_{i=1}^N y_i^{\xi_i}  \frac{\Gamma(2s)}{\Gamma(2s+\xi_i)}\right]
\mu_{N}(dy)
=\sum_{n=0}^{|{ \xi}|}T_R^{|{ \xi}|-n}(T_L-T_R)^n g_{\xi}(n)
\end{equation} 
where $|\xi| = \sum_{i=1}^N \xi_i$ and 
\begin{equation}\label{eq:gfunc}
 g_{\xi}(n)=
\sum_{{\underset{n_1+\ldots+n_{\scalebox{0.5} N}=n}{(n_1,\ldots,n_{\scalebox{0.5} N}) \in \N_0^{\scalebox{0.5} N}}}}
\prod_{i=1}^N  {\binom{\xi_i}{n_i}} \prod_{j=1}^{n_i}\frac{2s(N+1-i)-j+\sum_{k=i}^{N}n_k}{2s(N+1)-j+\sum_{k=i}^{N}n_k}\,.
\end{equation}

\item Equivalently, interpreting the multi-index $\xi=\sum_{i=1}^N \xi_i \delta_i$ as a configuration of (dual) particles  located at positions $(x_k)_{1\le k \le |\xi|}$ with $x_1\le x_2\le \ldots \le  x_{|\xi|}$, we have
\begin{equation}
\label{ed}
\int_{\R_+^N}\left[\prod_{i=1}^N y_i^{\xi_i} \frac{\Gamma(2s)}{\Gamma(2s+\xi_i)}\right]
\mu_{N}(dy)
= \sum_{n=0}^{|\xi|}T_R^{|  \xi|-n}(T_L-T_R)^n\
g_{x}(n)
\end{equation} 
with
\begin{equation}
\label{eq: express-two}
 g_{x}(n)=
\sum_{1\leq i_1< \ldots< i_n\leq |  \xi|} \; \prod_{\alpha=1}^n\frac{n-\alpha+2s(N+1-x_{i_\alpha})}{n-\alpha+2s(N+1)}\,.
\end{equation} 

\end{itemize}

\end{theorem}

\begin{remark}[Recovering the equilibrium case]
If one chooses equal reservoir parameters  in Theorem~\ref{theo:main} then one finds the moments of the equilibrium stationary measure of Theorem~\ref{prop:eq}.
Indeed, setting $T_L = T_R = T$, only the term $n=0$
survives in the right hand side of \eqref{eq:moments} yielding
$$
\bigint_{\R_+^N}\left[\prod_{i=1}^N y_i^{\xi_i}  \frac{\Gamma(2s)}{\Gamma(2s+\xi_i)}\right]
\mu_{N}(dy)
=T^{|\xi|}.
$$
\end{remark}

\begin{remark}[Long-range correlations]
Considering the dual configuration with
just one dual particle at site $i$
one deduces a microscopic linear
profile
\be
\label{eq:linear-micro}
\int_{\R_+^N} y_i \mu_N(dy) = 2s\Big(T_L + \frac{i}{N+1}(T_R -T_L)\Big) \;.
\ee
The expression for the multivariate
moments given in  \eqref{ed},
  \eqref {eq: express-two} implies
  long-range correlations.
  For instance the covariance between the energies at two sites $i<j$
read
\begin{eqnarray}
\label{eq:cov-micro}
&&\int_{\R_+^N} y_i y_j \mu_N(dy) - 
\Big(\int_{\R_+^N} y_i \mu_N(dy)\Big)
\Big(\int_{\R_+^N} y_j \mu_N(dy)\Big)
\\
&& \qquad\qquad\qquad\qquad = 
(2s)^2(T_L-T_R)^2 \frac{i}{(N+1)^2} \frac{N+1-j}{(2s(N+1)+1)} \nonumber
\end{eqnarray}
More generally, considering
the cumulants $\kappa_n$
of $n\ge 2$, the energies at the microscopic points $i_1<\ldots < i_n$, for large $N$ behave as 
\be
\lim_{N\to\infty}N^{n-1}\kappa_n = f_n(u_1,\ldots,u_n)(T_R-T_L)^n\,,
\ee
where $i_k = \lfloor N u_k \rfloor$ and appropriate functions $f_n$ (the first being
$
f_2(u_1,u_2) = (2s)u_1(1-u_2)
)$.

\end{remark}

Although Theorem~\ref{theo:main}
shows that in non-equilibrium 
the stationary measure is not
a Gibbs measure, it is possible to show that in the thermodynamic limit $N\to\infty$ 
the non-equilibrium stationary measure approaches {\em locally} a Gibbs distribution, i.e. the 
integrable heat conduction model
satisfies {\em local equilibrium}. 
Furthermore, transport of energy across the system satisfies Fourier's law.

To formalize this, let $O$ be the algebra of cylindrical {bounded continuous functions on ${\mathbb{R}_{+}^\mathbb{N}}$, i.e. a function which depends on the configuration $y$ only through a finite number of variables $y_i \in \mathbb{R}_+$.}
Denote by $\tau_i$ the translation by $i$, i.e. for all function $f\in O$
define $(\tau_i f)(j) = f(i+j)$.
For a measure $\mu$, we denote
by $\mathbb{E}_{\mu}(\cdot)$ the
expectation under this measure.

\begin{theorem}[Local equilibrium \& Fourier's law] 
\label{fourier}
Let $T_L,T_R$ be defined by 
\eqref{eq:rho} and let $\mu_N$ be the unique invariant  measure of the process generated by \eqref{eq:gen}.
Then the following holds: 
\begin{itemize}
\item[(i)]
For $u\in(0,1)$
\label{theo:loc}
\be
\label{theo:loc33}
\lim_{N\to\infty} \mathbb{E}_{\mu_N}(\tau_{[u N]} f) = \mathbb{E}_{\nu_{\lambda(u)}}(f) \qquad\qquad \forall \, f \in O
\ee
where $[x]$ denotes the integer part of $x\in\R$,
$\lambda(u) = 1/T(u)$ and
\be
\label{eq:linear}
T(u) = T_L + (T_R-T_L) u \,,
\ee
and
$\nu_{\lambda}$ is the product measure supported on $\R_+^{\N}$ with marginals
given by a Gamma distributions with shape parameter $2s>0$
and rate parameter $\lambda>0$
\be \label{eq:density22}
\nu_{\lambda}(dy) = \prod_{i=1}^{\infty} 
\frac{\lambda^{2s}}{\Gamma(2s)} y_i^{2s-1}e^{-\lambda y_i} dy_i\;.
\ee
\item[(ii)]
Define the stationary energy flux between
two neighbor sites $i,i+1$ by
\be
\label{eq:currbond}
J_{i,i+1} =  \int_{{\mathbb{R}_{+}^N}} \mu_N(dy)\, \big(y_i - y_{i+1}\big),
\ee
and the total stationary current in the thermodynamic limit as
\be
J = \lim_{N\to\infty} \sum_{i=1}^{N-1} J_{i,i+1}.
\ee
Then Fourier's law holds, namely
\be
\label{eq:fourier}
J = - K_s\frac{dT(u)}{du},  \qquad\qquad u\in(0,1),
\ee
where the conductivity $K_s = 2s$ and the macroscopic temperature profile $T(\cdot)$ is defined by \eqref{eq:linear}.
\end{itemize}
\end{theorem}

\section{Proofs of the results} \label{sec:proofs}
In this section we prove the theorems stated previously in Section~\ref{sec2}. The main
ingredient 
used is the 
duality relation which is discussed below. More precisely, we will show that the open heat conduction model is dual to a purely absorbing particle system. This dual process  is the same dual process obtained for the boundary-driven ``harmonic'' model  in \cite{frassek2021exact} but, as we will see,  with a different duality function.

 The absorbing dual process and the associated duality function
are identified in Section
\ref{sec3.1}. 
Once the duality result is established, we can  use the
expression for the absorption probabilities of the dual particles
that were obtained in \cite{frassek2021exact} to prove Theorem~\ref{theo:main}, see Section~\ref{sec3.3}.
The case of equilibrium, i.e. Theorem~\ref{prop:eq}, does
not require the explicit form
of the absorption probabilities
and is treated separately in Section~\ref{sec3.2}.

Theorem~\ref{fourier} about local equilibrium and Fourier's law will be proved in Section~\ref{sec3.4}.
Contrary to the usual situation of non-integrable models where one needs a coupling argument with independent particles, here
local equilibrium is obtained by a direct computation that uses the explicit form of the multi-point correlation function
of Theorem~\ref{theo:main}.

The fact that the heat conduction model studied in this paper and the boundary-driven `harmonic' model considered in \cite{frassek2021exact}
have the same dual process
can be explained by the fact
that the two models both arise
from the integrable open XXX Heisenberg spin chain when considering different representations of the non-compact
$\mathfrak{sl}(2)$ Lie algebra. This is discussed in Section~\ref{sec4}.

\subsection{The associated dual process}
\label{sec3.1}

We introduce the Markov process obtained in \cite{frassek2021exact} which is related via duality to the one of the heat conduction model in \eqref{eq:gen}.

\begin{definition}[Dual absorbing process]
\label{def:processdual}
For $s>0$, let the function $\varphi_s : \N \times \N \to \R_+$ 
\be
\label{eq:varphi}
\varphi_s(k,n) :=  \frac{1}{k}\frac{\Gamma (n+1) \Gamma (n-k+2 s)}{ \Gamma (n-k+1) \Gamma (n+2 s)} 
\mathds{1}_{\{1,2,\ldots,n\}}(k)
\ee
where $\mathds{1}_{A}$ is the indicator function
of the set $A$.
We consider the Markov chain $\{\xi(t) 
\,:\, t\ge 0\}$ on $\N^{N+2}$
whose generator ${\mathscr L}^{\mathrm{dual}}$,
acting on functions $f : \N^{N+2}\to\R$, is given by
\be
\label{eq:gendual}
{\mathscr L}^{\mathrm{dual}} =
{\mathscr L}^{\mathrm{dual}}_{0,1}
+
\sum_{i=1}^{N-1}{\mathscr L}^{\mathrm{dual}}_{i,i+1} 
+ 
{\mathscr L}^{\mathrm{dual}}_{N,N+1}\,,
\ee
where for $i \in \{1,2,\ldots N-1\}$ 
\begin{eqnarray}
\label{eq:genbulk}
({\mathscr L}^{\mathrm{dual}}
_{i,i+1} f)(\xi) 
& = & 
\sum_{k=1}^{\xi_i}\varphi_s(k,\xi_i) \Big[f(\xi-k \delta_i + k \delta_{i+1}) - f(\xi)\Big] \\
& + & 
\sum_{k=1}^{\xi_{i+1}} \varphi_s(k,\xi_{i+1}) \Big[f(\xi+k\delta_i - k \delta_{i+1}) - f(\xi)\Big]\,, \nonumber
\end{eqnarray}
while 
\begin{eqnarray}
\label{eq:gendualbdry1}
{\mathscr L}^{\mathrm{dual}}
_{0,1} f(\xi) 
& = &  
\sum_{k=1}^{\xi_1} \varphi_s(k,\xi_1) \Big[f(\xi-k\delta_1 + k \delta_0) - f(\xi)\Big]\,,
\end{eqnarray}
and
\begin{eqnarray}
\label{eq:gendualbdry2}
{\mathscr L}^{\mathrm{dual}}
_{N,N+1} f(\xi) 
& = &  
\sum_{k=1}^{\xi_N} \varphi_s(k,\xi_N) \Big[f(\xi-k\delta_N + k \delta_{N+1}) - f(\xi)\Big]\,.
\end{eqnarray}
\end{definition}
The Markov process $\{\xi(t) \,:\, t\ge 0\}$,
initialized from a configuration $\xi\in\N^{N+2}$ 
describes the motion of $|\xi| = \xi_0+\ldots + \xi_{N+1}$ indistinguishable particles which move on the lattice sites $\{0,\ldots, N+1\}$
and are absorbed at the boundary sites $0$ and $N+1$, i.e. they cannot reenter the chain with the lattice sites $\{1,\ldots, N\}$. Eventually, all the $|\xi|$ particles  
get absorbed at the boundary sites.

\begin{theorem}[Duality] \label{duality-thm}
Let $D: \R_+^N\times \N^{N+2}$ be the duality function defined by
\be\label{dualityfunction}
D(y,\xi) =\left( \dfrac{1}{\lambda_L}\right) ^{\xi_0} \left(\prod_{i=1}^N d(y_i,\xi_i) \right) \left( \dfrac{1}{\lambda_R}\right) ^{\xi_{N+1}} 
\ee
with
\be
d(y_i,\xi_i) = y_i^{\xi_i}  \frac{\Gamma(2s)}{\Gamma(2s+\xi_i)}.
\ee
Then for every $t\ge 0$ and for all $(y,\xi)\in\R_+^N\times \N^{N+2}$ one has
\be
\mathbb{E}_{y}\Big[D\Big(y(t),\xi\Big)\Big]
=
\mathbb{E}_{\xi}\Big[D\Big(y,\xi(t)\Big)\Big],
\ee
where on the left hand side the expectation is 
w.r.t. the process $\{y(t)\,,\,t\ge 0\}$
initialized from $y$ and on the right hand side
the expectation is w.r.t. the process $\{\xi(t)\,,\,t\ge 0\}$
initialized from $\xi$.

\end{theorem}

\begin{proof}
It is enough to prove that, for all $y\in \R_+^N$ and $\xi\in\N^{N+2}$, it holds
\begin{equation}
\label{duality}
\Big({\mathscr{L}}D(\cdot,\xi)\Big)(y) = 
\Big({\mathscr{L}}^{\mathrm{dual}}D(y,\cdot)\Big)(\xi).
\end{equation}
To prove this we treat  each term
appearing in the sum defining the generators separately. 
Let's start with the bulk part and  consider the local generator acting on the bond $(i,i+1)$. Its right action  on the duality function \eqref{dualityfunction} is
\begin{align}
\label{bulk-massage}
& \Big({\mathscr{L}^{\rightarrow}_{i,i+1}}D(\cdot , \xi)\Big)(y) =  T_L^{\xi_0} \left( \prod_{j\neq \{i,i+1\}} d(y_j,\xi_j)\right) T_R^{\xi_{N+1}} \cdot\\&
\qquad\qquad \int_0^{y_i} \dfrac{d\alpha}{\alpha} \left(1-\frac{\alpha}{y_i}\right)^{2s-1}\left[ \left( y_{i} - \alpha \right)^{\xi_{i}}  \left( y_{i+1} + \alpha \right)^{\xi_{i+1}} - y_{i}^{\xi_{i}}y_{i+1}^{\xi_{i+1}}  \right] \dfrac{\Gamma(2s)\Gamma(2s)}{\Gamma(\xi_{i} + 2s) \Gamma(\xi_{i+1} + 2s)}.\nonumber
\end{align}
We focus on the integral on the right hand side of the above display.
Using the change of variable $\alpha \to \alpha/y_i$, the integral can be rewritten as
$$
\int_0^{1} \dfrac{d\alpha}{\alpha} \left(1-\alpha\right)^{2s-1}\left[ \left( y_{i} - \alpha y_i \right)^{\xi_{i}}  \left( y_{i+1} + \alpha y_i \right)^{\xi_{i+1}} - y_{i}^{\xi_{i}}y_{i+1}^{\xi_{i+1}}  \right].
$$
Next, expanding the term {$ \left( y_{i+1} + \alpha y_{i} \right)^{\xi_{i+1}} $ with the Newton binomial we get
$$
\int_{0}^{1}\dfrac{d\alpha}{\alpha} (1-\alpha)^{2s-1} \Biggl(  \ 
 \sum_{\ell=0}^{\xi_{i+1}}\binom{\xi_{i+1}}{\ell} y_{i}^{\xi_i+\ell} y_{i+1}^{\xi_{i+1} - \ell} \alpha^{\ell} (1-\alpha)^{\xi_i}
- y_{i}^{\xi_i} y_{i+1}^{\xi_{i+1}}
 \Biggl) \,.
$$
}
Furthermore, using the integral representation of the Beta function 
\be\label{eq:intbeta}
B(a,b) =   \int_{0}^{1}  \alpha^{a-1} (1- \alpha)^{b-1} d\alpha
\ee 
and isolating the term $\ell=0$ in the summation
we arrive to
\begin{align*}
& 
\int_0^{y_i} \dfrac{d\alpha}{\alpha} \left(1-\frac{\alpha}{y_i}\right)^{2s-1}\left[ \left( y_{i} - \alpha \right)^{\xi_{i}}  \left( y_{i+1} + \alpha \right)^{\xi_{i+1}} - y_{i}^{\xi_{i}}y_{i+1}^{\xi_{i+1}}  \right] = \\
&  
\qquad\qquad\Biggl( \, \sum_{\ell=1}^{\xi_{i+1}}  \binom{\xi_{i+1}}{\ell} y_{i}^{\xi_{i} + \ell} y_{i+1}^{\xi_{i+1} - \ell}   B(\ell,\xi_{i}+2s) \Biggl) \,+\,  
y_{i}^{\xi_i} y_{i+1}^{\xi_{i+1}}  \int_{0}^{1}  d\alpha\, \alpha^{-1} (1-\alpha)^{2s-1} \left[ (1 - \alpha)^{\xi_i}  -1  \right] .    
\end{align*}
Inserting this into \eqref{bulk-massage}
we get
\be
\label{I+II}
\Big({\mathscr{L}^{\rightarrow}_{i,i+1}}D(\cdot , \xi)\Big)(y) =
\text{I} + \text{II}
\ee
where
$$
\text{I} = T_L^{\xi_0} \left( \prod_{j\neq \{i,i+1\}} d(y_j,\xi_j)\right) T_R^{\xi_{N+1}} \cdot\Biggl(\sum_{\ell=1}^{\xi_{i+1}}  \binom{\xi_{i+1}}{\ell} y_{i}^{\xi_{i} + \ell} y_{i+1}^{\xi_{i+1} - \ell}   B(\ell,\xi_{i}+2s) \Biggl)
 \dfrac{\Gamma(2s)\Gamma(2s)}{\Gamma(\xi_{i} + 2s) \Gamma(\xi_{i+1} + 2s)}
 $$
 and
 $$
\text{II} = T_L^{\xi_0} \left( \prod_{j\neq \{i,i+1\}} d(y_j,\xi_j)\right) T_R^{\xi_{N+1}} \cdot\Biggl(y_{i}^{\xi_i} y_{i+1}^{\xi_{i+1}}  \int_{0}^{1}  d\alpha\, \alpha^{-1} (1-\alpha)^{2s-1} \left[ (1 - \alpha)^{\xi_i}  -1  \right] \Biggl)
 \dfrac{\Gamma(2s)\Gamma(2s)}{\Gamma(\xi_{i} + 2s) \Gamma(\xi_{i+1} + 2s)}.
 $$
The first term $\text{I}$ can be written in terms of Gamma functions as
\begin{align*}
& \text{I} = T_L^{\xi_0} \left( \prod_{j\neq \{i,i+1\}} d(y_j,\xi_j)\right) T_R^{\xi_{N+1}} \cdot \\
&\qquad\qquad \sum_{\ell=1}^{\xi_{i+1}} \dfrac{\Gamma(\xi_{i+1} +1)}{\Gamma(\ell +1) \Gamma(\xi_{i+1} - \ell +1)} y_{i}^{\xi_{i} + \ell} y_{i+1}^{\xi_{i+1} - \ell}   \dfrac{\Gamma(\ell) \Gamma(\xi_{i} + 2s)}{\Gamma(\xi_{i} + \ell +2s)} \dfrac{\Gamma(2s)\Gamma(2s)}{\Gamma(\xi_{i} + 2s) \Gamma(\xi_{i+1} + 2s)}.
\end{align*}
Recalling the definition \eqref{dualityfunction} 
of the duality function 
and the definition \eqref{eq:varphi} of the function $\varphi_s(k,n)$ 
one finds
\be
\label{I}
\text{I} =  
\sum_{\ell=1}^{\xi_{i+1}}  \varphi_s (\ell, \xi_{i+1}) D(y; \xi + \ell \delta_{i} -\ell \delta_{i+1}).
\ee
For the second term $\text{II}$, we observe that
\be
\label{use}
\int_{0}^{1}  d\alpha\, \alpha^{-1} (1-\alpha)^{2s-1} \left[ (1 - \alpha)^{\xi_i}  -1  \right] = \psi(2s)-\psi(\xi_i+2s)
\ee
where $\psi$ indicates the digamma function (i.e. the logaritmic derivative of the Gamma function). Here we  used  integral representation of the digamma function
\begin{equation}
\label{eq:intdigamma}
    \psi(z+1)=-\gamma_e+\int_0^1 d\beta\,\frac{1-\beta^z}{1-\beta}\,,
\end{equation} 
with $\gamma_e$ the Euler–Mascheroni constant.
Therefore, recalling again the definition \eqref{dualityfunction} 
of the duality function,  we get
\begin{align}
\label{II-simplify}
\text{II} =  
\Big(\psi(2s)-\psi(\xi_i+2s)\Big) D(y; \xi).
\end{align}
We then proceed by observing that 
\begin{equation}\label{digammadifference}
    \psi(2s)-\psi(\xi_i+2s) = - \sum_{\ell = 1}^{\xi_i} \varphi_s(\ell, \xi_i).
    \end{equation}
This identity can be shown by writing the rates \eqref{eq:varphi} in terms of the Beta function
\be 
\varphi_s(k,n) =  \binom{n}{k}B(k,n-k+2s)
\ee
and  using  the integral representation of the Beta function \eqref{eq:intbeta} and of the digamma function \eqref{eq:intdigamma} we find
\begin{equation}
 \begin{split}
\sum_{k=1}^n\varphi_s(k,n) &=\int_0^1d t \, \frac{t^{n+2s-1}}{1-t}\sum_{k=1}^n \binom{n}{k}\left(\frac{1-t}{t}\right)^{k}=\int_0^1d t \, \frac{t^{2s-1}}{1-t}(1-t^{n})= \psi(n+2s)-\psi(2s).\nonumber
 \end{split}
\end{equation}  
Thus, combining together  \eqref{II-simplify}  
and \eqref{digammadifference} we arrive to
\be
\label{II}
\text{II} = - \sum_{\ell=1}^{\xi_{i}}  \varphi_s (\ell, \xi_{i}) D(y; \xi).
\ee
Inserting the expressions \eqref{I} and \eqref{II} 
into the right hand side of  \eqref{I+II},
we get
$$
\Big({\mathscr{L}^{\rightarrow}_{i,i+1}}D(\cdot , \xi)\Big)(y) = \sum_{\ell=1}^{\xi_{i+1}}  \varphi_s (\ell, \xi_{i+1}) D(y; \xi + \ell \delta_{i} -\ell \delta_{i+1}) - \sum_{\ell=1}^{\xi_{i}}  \varphi_s (\ell, \xi_{i}) D(y; \xi).
$$
If we now define 
\be\label{def:hs}
h_s(n) = \sum_{k=1}^n \varphi_s (k,n),
\ee
and we recall Definition \ref{def:processdual} of the dual process, then we have shown that 
\begin{equation*}
\Big({\mathscr{L}^{\rightarrow}_{i,i+1}}D(\cdot , \xi)\Big)(y) = 
\Big({\mathscr{L}^{\leftarrow,\text{dual}}_{i,i+1}}D(y, \cdot  )\Big)(\xi) + 
\Big(h_s(\xi_{i+1}) - h_s(\xi_i)\Big) D(y,\xi)\,,
\end{equation*}
{
where 
${\mathscr{L}^{\leftarrow,\text{dual}}_{i,i+1}} = \sum_{k=1}^{\xi_{i+1}} \varphi_s(k,\xi_{i+1}) \Big[f(\xi+k\delta_i - k \delta_{i+1}) - f(\xi)\Big] $.
}
Repeating the same computation for the left part of the bulk generator leads to 
    \begin{equation*}
\Big({\mathscr{L}^{\leftarrow}_{i,i+1}}D(\cdot , \xi)\Big)(y) =
\Big({\mathscr{L}^{\rightarrow,\text{dual}}_{i,i+1}}D(y, \cdot  )\Big)(\xi) + 
\Big(h_s(\xi_{i}) - h_s(\xi_{i+1})\Big) D(y,\xi)\,,
\end{equation*}
{
where 
${\mathscr{L}^{\rightarrow,\text{dual}}_{i,i+1}} = \sum_{k=1}^{\xi_i}\varphi_s(k,\xi_i) \Big[f(\xi-k \delta_i + k \delta_{i+1}) - f(\xi)\Big]$.
}
All in all, adding up the last two expressions one arrives to 
\be
\label{bulk-duality}
\Big({\mathscr{L}_{i,i+1}}D(\cdot , \xi)\Big)(y) = \Big({\mathscr{L}^{\text{dual}}_{i,i+1}}D(y, \cdot  )\Big)(\xi)\,,
\ee
which proves bulk duality.

We now consider the left reservoir. 
Spelling out the action of the generator $\mathscr{L}_{1}$ on the duality function
$D$ we get
\be
\label{use2}
\Big(\mathscr{L}_{1} D(\cdot , \xi)\Big)(y) = \text{III} + \text{IV}
\ee
where
\begin{align*}
\text{III} &= T_L^{\xi_0} \int_{0}^{y_1} \dfrac{d\alpha}{\alpha} \left( 1-\dfrac{\alpha}{y_1}\right)^{2s-1} 
\left[(y_1 - \alpha)^{\xi_1} - y_1^{\xi_1}\right]
\dfrac{\Gamma(2s)}{\Gamma(2s+ \xi_1)}  \left(\prod_{j=2}^{N}d(y_j,\xi_j)\right) \, T_R^{\xi_{N+1}}
 \end{align*}
 and
\begin{align*}
\text{IV} &=
 T_L^{\xi_0} \int_{0}^{\infty} \dfrac{d\alpha}{\alpha} e^{-\lambda_L \alpha} \left[(y_1 + \alpha)^{\xi_1} - y_1^{\xi_1}\right]
\dfrac{\Gamma(2s)}{\Gamma(2s+ \xi_1)}
 \left(\prod_{j=2}^{N}d(y_j,\xi_j)\right) \, T_R^{\xi_{N+1}}.
 \end{align*}
For the term III we change variable $\alpha \to  {\alpha}/{y_1}$ so that we get 
\begin{align*}
\text{III} &= T_L^{\xi_0} \int_{0}^{1} {d\alpha}{\alpha}^{-1} \left( 1-\alpha\right)^{2s-1} 
\left[(1 - \alpha)^{\xi_1} -1\right] y_1^{\xi_1}
\dfrac{\Gamma(2s)}{\Gamma(2s+ \xi_1)}  \left(\prod_{j=2}^{N}d(y_j,\xi_j)\right) \, T_R^{\xi_{N+1}}.
 \end{align*}
 By using \eqref{use}  and \eqref{digammadifference} we recognize that
 \be
 \label{III}
 \text{III} = - \sum_{\ell=1}^{\xi_1} \varphi_{s}(\ell,\xi_1) D(y,\xi).
 \ee
 For the term IV we write the Newton binomial of $(y_1 + \alpha)^{\xi_1}$ and we note that the first term of the sum cancel to get
 \begin{align*}
\text{IV} &=
 T_L^{\xi_0} \sum_{\ell=1}^{\xi_1} {\xi_1 \choose \ell} y_1^{\xi_1-l}
 \frac{\Gamma(2s)}{\Gamma(2s+ \xi_1)}
 \int_{0}^{\infty} {d\alpha}{\alpha}^{\ell-1} e^{-\lambda_L \alpha} 
 \left(\prod_{j=2}^{N}d(y_j,\xi_j)\right) \, T_R^{\xi_{N+1}}.
 \end{align*}
The integral gives $\left(T_L\right)^{\ell} \Gamma(\ell) $ so that we can reconstruct the function $\varphi_s(\ell,\xi_1)$ and get 
 \begin{align*}
\text{IV} &=
  \sum_{\ell=1}^{\xi_1}\varphi_s(\ell,\xi_1)T_L^{\xi_0+\ell}  y_1^{\xi_1-\ell}
 \frac{\Gamma(2s)}{\Gamma(2s+ \xi_1 -\ell)}
 \left(\prod_{j=2}^{N}d(y_j,\xi_j)\right) \, T_R^{\xi_{N+1}}.
 \end{align*}
 Upon recalling the definition of the duality function \eqref{dualityfunction}, this gives
\be
 \label{IV}
\text{IV} =
  \sum_{\ell=1}^{\xi_1}\varphi_s(\ell,\xi_1)
  D(y,\xi-\ell \delta_{1}+ \ell \delta_0)\,.
\ee
Inserting the expressions \eqref{III} and \eqref{IV} into the right hand side of \eqref{use2} we then find
\be
\label{left-duality}
\Big(\mathscr{L}_{1} D(\cdot , \xi)\Big)(y) 
=   \Big({\mathscr{L}^{\text{dual}}_{0,1}}D(y, \cdot)\Big)(\xi) \,,
\ee
which proves the left boundary duality.
Similarly one gets
\be
\label{right-duality}
\Big(\mathscr{L}_{N} D(\cdot , \xi)\Big)(y) 
=   \Big({\mathscr{L}^{\text{dual}}_{N,N+1}}D(y, \cdot)\Big)(\xi) \,,
\ee
which establishes the right boundary duality.
The combination of \eqref{bulk-duality}, \eqref{left-duality} and \eqref{right-duality}
implies \eqref{duality} and the proof of the theorem is  completed.
\end{proof}


    
    
 
 \subsection{Proof of Theorem~\ref{prop:eq}}
 \label{sec3.2}
 {The moments of the stationary measure can be studied via the duality relation.}
The following statement (see Proposition 1 and 2 of \cite{giardina2007duality}) is a classical consequence of the duality relation: it holds in and out of equilibrium and it allows to write the 
expectation (with respect to the stationary measure) of the duality function in terms of the dual process. The stationary expectation of the duality function can be written in terms of the absorption probabilities of the dual process, namely 
\begin{equation}
\label{use-duality}
\mathbb{E}_{\mu_{N}}\left[ D(y, \xi) \right]= \sum_{k=0}^{|\xi|} \left(\dfrac{1}{\lambda_{L}} \right)^{k}  \left(\dfrac{1}{\lambda_{R}} \right)^{|\xi| - k} p_{\xi}(k) \,.
\end{equation}
Above $p_{\xi}(k) = \mathbb{P}_{\xi} \Big(\xi(\infty) = k \delta_{0} + (|\xi| - k)\delta_{N+1}\Big)$ denotes the probability of $k$ dual particles being eventually absorbed at site $0$ and the remaining $|\xi| - k$ particles being eventually absorbed at site $N+1$, when the dual process is initialized from the configuration $\xi\in\N^{N+2}$ with $|\xi| = \sum_{i=1}^{N} \xi_i $ particles and no particles at the sites $\{0,N+1\}$.

Under equilibrium  $\lambda_L = \lambda_R = \lambda$ we get
 \begin{equation*}
\mathbb{E}_{\mu_{N}} \left[D(y, \xi)\right]=\left(  \dfrac{1}{\lambda} \right)^{|\xi|}\,.
\end{equation*}
Using the explicit form of the duality function in equation \eqref{dualityfunction}, the above display becomes
 \begin{equation*}
\mathbb{E}_{\mu_{N}} \left[\prod_{i=1}^{N} y_{i}^{\xi_{i}} \right]= \prod_{i=1}^{N} \left(  \dfrac{1}{\lambda} \right)^{\xi_{i}} \dfrac{\Gamma(2s + \xi_i)}{\Gamma(2s)} \;,
\end{equation*}
which are recognized as the multivariate moments
of the product probability measure in equation \eqref{eq:density}.
{
In Appendix~\ref{appendixa} we show that $\mu_N$ is also reversible.}

  \subsection{Proof of Theorem~\ref{theo:main}}
  \label{sec3.3}
  
  The proof starts again from equation \eqref{use-duality}, which expresses the expectation of the duality function computed in a configuration $\xi$ with $n$ dual particles as a polynomial of order $n$ in the two temperatures $T_L$ and $T_R$,
  whose coefficient are the absorption probabilities of the $n$ dual particles.
  As our dual process is the same as the one in \cite{frassek2021exact} (cf. Definition \ref{def:processdual}
  of this paper to Definition 2.3
  of \cite{frassek2021exact}),
  we can just use the expression
  for the absorption probabilities found there.  { The formula for the  absorption probabilities \cite[(2.50)]{frassek2021exact} has been obtained by two of the authors, see also \cite{Frassek:2019imp,Frassek:2020omo} for the case of SSEP, exploiting the integrable structure of the model and its algebraic formulation within the Quantum Inverse Scattering Method \cite{Kulish:1981gi,Sklyanin:1988yz}, see also \cite{Faddeev:1996iy} for an excellent review.}
  
  Thus we obtain  the multivariate
  moments given in \eqref{eq:moments}, \eqref{eq:gfunc}  when
  the configuration having $n$ dual particles is described by the occupation numbers, or the expression given in \eqref{ed},
  \eqref {eq: express-two} when
  the configuration  is described by assigning
  the ordered positions of the $n$ dual particles.

   \subsection{Proof of 
  Theorem~\ref{fourier}}
  \label{sec3.4}

Let $u\in(0,1)$.
To prove the first item (local equilibrium) it is enough to study the convergence of the moments
in a macroscopic point $\lfloor uN\rfloor $.
Considering a configuration $(\xi_1,\ldots,\xi_N)\in\N^N$ having
$|\xi|$ dual particles located at positions $1\le x_1 \le x_2 \le \ldots \le x_{|{ \xi}|}\le N$, 
we define the shifted configuration $(\xi_1^u,\ldots,\xi_N^u)\in\N^N$ having
the particles located at positions
$x_1^u \le x_2^u \le \ldots \le x^u_{|{ \xi}|}$ with $x_k^u = x_k + [uN]$.
Then we would like to prove that 
\be
\label{proof-local}
\lim_{N\to\infty} \int_{\R_+^N} \Big[\prod_{i=1}^N y_i^{\xi_i^u}\frac{\Gamma(2s)}{\Gamma(2s+\xi_i^u)}\Big] = [T(u)]^{|\xi|}
\ee
where $T(u)$ denotes the macroscopic linear profile \eqref{eq:linear}.
Formula \eqref{eq: express-two}
gives that
\begin{equation}
\lim_{N\to\infty} g_{ x^u}(n)
= 
\lim_{N\to\infty}  \sum_{1\leq i_1< \ldots< i_n\leq |  \xi|} \; \prod_{\alpha=1}^n\frac{n-\alpha+2s(N+1-x_{i_\alpha}-\lfloor uN\rfloor)}{n-\alpha+2s(N+1)}=
{|\xi| \choose n} (1-u)^n\,.
\end{equation} 
Using then 
\eqref{ed}
we get
\begin{eqnarray}
\lim_{N\to\infty} \int_{\R_+^N} \Big[\prod_{i=1}^N y_i^{\xi_i^u}\frac{\Gamma(2s)}{\Gamma(2s+\xi_i^u)}\Big] 
& = &
\sum_{n=0}^{|\xi|}T_R^{|  \xi|-n}(T_L-T_R)^n {|\xi| \choose n} (1-u)^n \nonumber \\
& = & 
 [T_R + (T_L-T_R)(1-u)]^{|\xi|}\,.
\end{eqnarray}
This proves \eqref{proof-local}, which in  turn implies \eqref{theo:loc33} {when $f$ is a polynomials functions}.

To prove the second item (Fourier's law) we observe that,
as a consequence of the linear microscopic profile \eqref{eq:linear-micro},  the average current is the same for all bonds and is given by
 \be
 J_{i,i+1} = - 2s \frac{T_R-T_L}{N+1}.
 \ee
From this,  \eqref{eq:fourier} immediately follows.

\section{Algebraic description}
\label{sec4}

In this section we show how the integrable heat conduction model \eqref{eq:gen} can be expressed algebraically in terms of the generators of the $\mathfrak{sl}(2)$ Lie algebra. 
More precisely,  we establish a relation between the Markov generator
of the integrable heat conduction model and the Hamiltonian of the  open XXX Heisenberg spin chain, in a particular
representation, by using the
algebra generators as building blocks. 

Once the algebraic description of the Levy generator of the heat conduction model via the $\mathfrak{sl}(2)$ algebra generators is obtained, the proof of duality is a consequence of a basic intertwining relation between two representations of the underlying algebra.
In other words, the dual absorbing harmonic process also arises from the open XXX  chain with non-compact spins, using a different representation.
Duality is essentially the statement that the two representations (one
leading to the boundary-driven Levy process and the other leading to the absorbing harmonic process)
are indeed equivalent
representations.

The identification of the algebraic expression of the Levy generator closely follows the techniques used in the work of Derkachov \cite{Derkachov:1999pz} in relation to Baxter Q-operators for the XXX Heisenberg spin chain.

\subsection{Levy generator}
\label{algebraic-original}
For the description of the Levy generator we consider the following representation (labeled by the parameter $s>0$) of the $\mathfrak{sl}(2)$ Lie algebra
\begin{equation}\label{ctssu11} 
\begin{split}
\mathscr{K}^{+}= y\,,\qquad 
\mathscr{K}^{-} =(y \partial_{y} + 2s)\partial_{y} \,,\qquad 
\mathscr{K}^{0} = y\partial_{y} +s \,.
\end{split}
\end{equation}
The operators $\mathscr{K}^{+}, \mathscr{K}^{-}, \mathscr{K}^{0}$, acting on polynomial functions,
generate highest weight state representations of the $\mathfrak{sl}(2)$ Lie algebra and  satisfy the commutation relations:
\be  \label{crelation}
 [{\mathscr{K}}^{0},{\mathscr{K}}^{\pm}]=\pm {\mathscr{K}}^{\pm} \qquad  \text{and} \qquad [{\mathscr{K}}^{+},{\mathscr{K}}^{-}]=-2{\mathscr{K}}^{0}\;.
\ee
At each lattice site $i \in \lbrace 1, \ldots, N \rbrace$ we consider a copy of the $\mathfrak{sl}(2)$ algebra. We write the site in the subscript of the generator $\mathscr{K}_i^{a}$ with $a\in\{+,-,0\}$. Generators at different sites commute.

We treat the bulk and boundary parts of the Markov generator of our model  in the two subsections below. We will show that the local generators in equation \eqref{eq:gen-bulk-right} and \eqref{eq:gen-bulk-left} can be written as

\begin{align} \label{abstract-right-bulk}
{\mathscr{L}^{\rightarrow}_{i,i+1}}  &=
- e^{\mathscr{K}^{+}_{i}(\mathscr{K}^{0}_{i+1}+s)^{-1}\mathscr{K}^{-}_{i+1}} (\psi(\mathscr{K}^{0}_{i} +s)-\psi(2s))
e^{- \mathscr{K}^{+}_{i}(\mathscr{K}^{0}_{i+1}+s)^{-1}\mathscr{K}^{-}_{i+1}} \\
\intertext{and similarly} 
 \label{abstract-left-bulk}
{\mathscr{L}^{\leftarrow}_{i,i+1}} &=
- e^{\mathscr{K}^{+}_{i+1}(\mathscr{K}^{0}_{i}+s)^{-1}\mathscr{K}^{-}_{i}} (\psi(\mathscr{K}^{0}_{i+1} +s)-\psi(2s))
e^{- \mathscr{K}^{+}_{i+1}(\mathscr{K}^{0}_{i}+s)^{-1}\mathscr{K}^{-}_{i}}\;.
\end{align}
For the boundary terms we will need to consider an additional representation
of the $\mathfrak{sl}(2)$ Lie algebra, which will be associated to  fictitious extra-sites $0$ and $N+1$
\be
\label{ctssu11-site0} 
\begin{split}
 \mathscr{S}^{+}_0= T_L(T_L \partial_{T_L} + 2s)\,,\qquad 
 \mathscr{S}^{-}_0 =\partial_{T_L} \,,\qquad 
 \mathscr{S}^{0}_0 = T_L\partial_{T_L} +s \,.
 \end{split}
\ee
\be
\label{ctssu11-siteN+1} 
\begin{split}
 \mathscr{S}^{+}_{N+1}= T_R(T_R \partial_{T_R} + 2s)\,,\qquad 
 \mathscr{S}^{-}_{N+1} =\partial_{T_R} \,,\qquad 
 \mathscr{S}^{0}_{N+1} = T_R\partial_{T_R} +s \,.
 \end{split}
\ee
The operators $\mathscr{S}_i^{+}, \mathscr{S}_i^{-}, \mathscr{S}_i^{0}$, with $i\in\{0,N+1\}$, acting on polynomial functions of  variable $T_L$ when $i=0$
and of  variable $T_R$ when $i=N+1$,
also satisfy the $\mathfrak{sl}(2)$
commutation relations:
\be 
 [{\mathscr{S}}_i^{0},{\mathscr{S}}_i^{\pm}]=\pm {\mathscr{S}}_i^{\pm} \qquad  \text{and} \qquad [{\mathscr{S}}_i^{+},{\mathscr{S}}_i^{-}]=-2{\mathscr{S}}_i^{0}\;.
\ee
Then, for the boundary terms
\eqref{eq:leftbnd} and \eqref{eq:rightbnd}, we will show that
they can be written as
\begin{equation} \label{abstract-boundary} 
\mathscr{L}_{1}
=
-  e^{- \mathscr{S}^{+}_{0}(\mathscr{S}^{0}_{0}+s)^{-1} \mathscr{K}^{-}_{1}} \left(\psi(\mathscr{K}^{0}_{1}+s)-\psi(2s)\right)
e^{ \mathscr{S}^{+}_{0}(\mathscr{S}^{0}_{0}+s)^{-1}  \mathscr{K}^{-}_{1}} \,,
 \end{equation}  
and 
 \begin{equation} \label{abstract-boundaryN} 
\mathscr{L}_{N}
=
-  e^{- \mathscr{K}^{-}_{N}\mathscr{S}^{+}_{N+1}(\mathscr{S}^{0}_{N+1}+s)^{-1} } \left(\psi(\mathscr{K}^{0}_{N}+s)-\psi(2s)\right)
e^{\mathscr{K}^{-}_{N}\mathscr{S}^{+}_{N+1} (\mathscr{S}^{0}_{N+1}+s)^{-1} } \,
 \end{equation}   
 where  $ \mathscr{S}^{+}_{0}(\mathscr{S}^{0}_{0}+s)^{-1}=T_L$
 and 
 $\mathscr{S}^{+}_{N+1}(\mathscr{S}^{0}_{N+1}+s)^{-1} = T_R$.

\subsubsection{Bulk generator}
To show the equivalence of the bulk Levy generator as given in \eqref{eq:gen-bulk-right} and \eqref{eq:gen-bulk-left} and the algebraic expressions \eqref{abstract-right-bulk} and \eqref{abstract-left-bulk} we first insert the algebra generators \eqref{ctssu11}   into the latter and obtain
\begin{equation}\label{eq:term1}
{\mathscr{L}}_{i,i+1}^{\rightarrow} 
 = 
- e^{y_i\partial_{i+1}} (\psi(y_i\partial_i+2s)-\psi(2s))
e^{-y_i\partial_{i+1}}
\end{equation}
and
\begin{equation}\label{eq:term2}
{\mathscr{L}}_{i,i+1}^{\leftarrow} 
 = 
- e^{y_{i+1}\partial_i} (\psi(y_{i+1}\partial_{i+1}+2s)-\psi(2s))
e^{-y_{i+1}\partial_i}\,.
\end{equation}
Here and in the following, we write in shorthand $\partial_i$ for the
partial derivative $\partial_{y_i}$.
Focussing on \eqref{eq:term1} and proceeding formally, using the integral representation \eqref{eq:intdigamma} of the digamma function, it can be written 
\begin{equation}
\begin{split}
 {\mathscr{L}}_{i,i+1}^{\rightarrow}&= 
 -\int_0^1 d\beta\,\frac{\beta^{2s-1}}{\beta-1}\left[e^{y_i\partial_{i+1}}\beta^{y_i\partial_i}e^{-y_i\partial_{i+1}}-1\right]\,.
 \end{split}
\end{equation} 
We can now determine the action of this operator on  functions. Noting that the exponential $e^{\pm y_i\partial_{i+1}}$ induces the shift $y_{i+1} \to y_{i+1} \pm y_i$ (see formula \eqref{B1}) and the operator
$\beta^{y_i\partial_i}$  induces the rescaling $y_i\to\beta y_i$ (see formula \eqref{B2}) we obtain
\begin{equation}
\begin{split}
{\mathscr{L}}_{i,i+1}^{\rightarrow}f(y)&=
- \int_0^1d\beta \frac{\beta^{2s-1}}{\beta-1}\left[f(y_1,\ldots,\beta y_i,y_{i+1} +(1-\beta)y_i,\ldots,y_N)-f(y)\right]\,.
\end{split}
\end{equation}  
Finally,  after a change of variables $\beta=1-\alpha y_i^{-1}$  we  find 
\begin{equation}
\begin{split}
{\mathscr{L}}_{i,i+1}^{\rightarrow}f(y)
 &=\int_0^{y_i}\frac{d\alpha}{\alpha} \left(1-\frac{\alpha}{y_i}\right)^{2s-1}\left[f(y_1,\ldots,y_i-\alpha,y_{i+1}+\alpha,\ldots,y_N)-f(y)\right]
\end{split}
\end{equation}  
which coincides with \eqref{eq:gen-bulk-right}. 
In the same way, for ${\mathscr{L}}_{i,i+1}^{\leftarrow}$ we derive  \eqref{eq:gen-bulk-left}  from \eqref{eq:term2}.

\subsubsection{Boundary generator}
We now turn to the boundary terms. We treat the left reservoir only, as the the right reservoir is treated analougsly. To obtain \eqref{eq:leftbnd} 
from the algebraic expressions \eqref{abstract-boundary}  
we again insert the algebra generators \eqref{ctssu11} for site $1$ and we use the algebra generators \eqref{ctssu11-site0}  for the extra site $0$.
Then we can write $\mathscr{L}_1$ as
\begin{equation}
\begin{split}
 \mathscr{L}_1&=-e^{-T_L(y_1\partial_1+2s) \partial_1}\left(\psi(y_1\partial_1+2s)-\psi(2s)\right)e^{T_L (y_1\partial_1+2s) \partial_1}\,.
 \end{split}
\end{equation} 
When expanding the exponentials, 
we obtain
\begin{equation}\label{eq:massagea}
\begin{split}
\mathscr{L}_1 
&=-\sum_{k,l=0}^\infty(-1)^l \left(\psi(y_1\partial_1+k+2s)-\psi(2s)\right)\frac{\Big(-T_L(y_1\partial_1+2s)  \partial_1\Big)^{k+l}}{k!l!} \end{split}
\end{equation} 
where we used the identity
$$
[(y_1\partial_1 + 2s)\partial_1]^k f(y_1\partial_1+2s) 
= 
f(y_1\partial_1+k+2s) 
[(y_1\partial_1 + 2s)\partial_1]^k
$$
{based on the commutation relations \eqref{crelation} for $ \mathscr{K}^{-}$ and $ \mathscr{K}^{0}$.}
Moving all derivatives to the right of the number operators we further get
\begin{equation}\label{eq:massageb}
\begin{split}
\mathscr{L}_1 
&=-\sum_{k,l=0}^\infty(-1)^l \left(\psi(y_1\partial_1+k+2s)-\psi(2s)\right)\frac{\Gamma(y_1\partial_1+2s+k+l)  }{\Gamma(y_1\partial_1+2s  )}\frac{(-T_L\partial_1)^{k+l}}{k!l!}\\
&=\sum_{m=0}^\infty\sum_{k=0}^m (-1)^{m-k}\left(\psi(y_1\partial_1+k+2s)-\psi(2s)\right)\frac{\Gamma(y_1\partial_1+2s+m)  }{\Gamma(y_1\partial_1+2s  )}\frac{(-T_L\partial_1)^{m}}{k!(m-k)!}\\
 &= \psi(y_1\partial_1+2s)-\psi(2s)+\log(1-T_L\partial_1) \,,
 \end{split}
\end{equation}
where in the final equality we used that 
\begin{equation}\label{eq:sumform1}
 \sum_{k=0}^m \frac{(-1)^{-k}}{k!(m-k)!}\left(\psi(y_1\partial_1+k+2s)-\psi(2s)\right)=\begin{cases}
 \psi(y_1\partial_1+2s)-\psi(2s)\qquad \text{for}\qquad m=0\\[5pt]
 -\frac{1}{m}\frac{\Gamma(y_1\partial_1+2s)  }{\Gamma(y_1\partial_1+2s+m  )}\qquad \text{for}\qquad m>0
             \end{cases}\,,
\end{equation} 
 cf.~\cite[(3.22)]{frassek2021exact}.
%
%
%
%

As before for the case of the bulk terms, we find that the first two terms in \eqref{eq:massageb} yield 
\begin{equation}
\label{eq:massage2}
\begin{split}
 \left(\psi(y_1\partial_1+2s)-\psi(2s)\right)f(y)
 &=-\int_0^{y_1}\frac{d\alpha}{\alpha} \left(1-\frac{\alpha}{y_1}\right)^{2s-1}\left(f(y_1-\alpha,\ldots,y_N)-f(y)\right)\,.
 \end{split}
\end{equation} 
The third term in the last line of \eqref{eq:massageb} yields
\begin{equation} 
\label{eq:massage3}
\log(1-T_L\partial_1)f(y)=
-\int_{0}^\infty\frac{d\alpha}{\alpha}\exp[-T_L^{-1}\alpha]\left( f(y_1+\alpha,\ldots,y_N)-f(y)\right)\,.
\end{equation} 
This can be seen by using the shift formula \eqref{B1} and
by writing
\begin{equation} 
\log(1-T_L\partial_1)
= -\log(T_L^{-1})+\log(T_L^{-1}-\partial_1)=
-\int_{0}^\infty\frac{d\alpha}{\alpha}\exp[-T_L^{-1}\alpha]\left(\exp[\alpha\partial_1]-1\right)\,.
\end{equation} 
Here we used the integral representation for $x>0$ 
\begin{equation}
 \log(x)=\int_0^\infty \frac{d\alpha}{\alpha}\left(e^{-\alpha}-e^{-x\alpha}\right)\,.
\end{equation} 
Combining together \eqref{eq:massageb}, \eqref{eq:massage2} and \eqref{eq:massage3} and using that $T_L = \lambda_L^{-1}$ we have thus shown that the algebraic expression \eqref{abstract-boundary} 
produces the boundary generator \eqref{eq:leftbnd}. 

Similar boundary operators appeared in the study of high energy QCD and  $\mathcal{N}=4$ super Yang-Mills theory
\cite{Lange:2003ff,Braun:2003wx,Belitsky:2019ygi}.

\subsection{Absorbing harmonic process}
\label{algebraic-dual}
For the description of the dual absorbing process we consider another representation, labelled by $s>0$, acting on functions $g:\N_0\to\R$  as
\begin{equation}\label{discretesu11} 
\begin{split}
\bar K^{+}g  (n)= (2s+n)g(n+1)\,,\qquad 
\bar K^{-}g (n)=ng(n-1)\,,\qquad 
\bar K^{0}g (n)=(n+s)g(n) \,,
 \end{split}
\end{equation}
where $g(-1)=0$. They satisfy the so-called ``dual'' (or ``conjugate'') $\mathfrak{sl}(2)$ Lie algebra with commutation relations
\be 
 [\bar{K}^{0}, \bar{K}^{\pm}]=\mp \bar{K}^{\pm} \qquad  \text{and} \qquad [\bar{K}^{+},\bar{K}^{-}]=2\bar{K}^{0}\;.
\ee

As before, treating the bulk and boundary generators separately we will show that the generator of the dual absorbing process in \eqref{eq:gendual} can be written in more abstract form using the representation \eqref{discretesu11}.
We note that this representation is closely related to (although different from) the one used in
cf.~\cite[(3.2)]{frassek2021exact}, see the remark below.

For the bulk part it will be convenient to split the action of the local generator in left and right part. As for the Levy generator we will show the following four identities: on one
hand for the bulk terms we have
\begin{align} \label{abstract-right-bulk-dis}
{\mathscr{L}^{\rightarrow,\text{dual}}_{i,i+1}}  &=
 - e^{-\bar{K}^{-}_{i}(\bar{K}^{0}_{i}+s)^{-1}\bar{K}^{+}_{i+1}} (\psi(\bar{K}^{0}_{i+1} +s)-\psi(2s))
e^{ \bar{K}^{-}_{i}(\bar{K}^{0}_{i}+s)^{-1}\bar{K}^{+}_{i+1}} \\
& \quad - h_s(\bar{K}^0_i-s) +h_s(\bar{K}^0_{i+1}-s)\,, \nonumber
\intertext{and similarly} 
 \label{abstract-left-bulk-dis}
{\mathscr{L}^{\leftarrow,\text{dual}}_{i,i+1}} &=
- e^{-\bar{K}^{+}_{i}\bar{K}^{-}_{i+1}(\bar{K}^{0}_{i+1}+s)^{-1}} (\psi(\bar{K}^{0}_{i} +s)-\psi(2s))
e^{ \bar{K}^{+}_{i}\bar{K}^{-}_{i+1}(\bar{K}^{0}_{i+1}+s)^{-1}}\\
& \quad + h_s(\bar{K}^0_i-s)-h_s(\bar{K}^0_{i+1}-s) \;.\nonumber
\end{align}
{
Here the function $h_s$ was defined in \eqref{def:hs} and we recall the action of the generators as given in \eqref{discretesu11}. }
Obviously, when taking the sum 
${\mathscr{L}^{\text{dual}}_{i,i+1}} = {\mathscr{L}^{\rightarrow,\text{dual}}_{i,i+1}} +{\mathscr{L}^{\leftarrow,\text{dual}}_{i,i+1}}$ in \eqref{eq:genbulk} the diagonal terms
related to the function $h_s$ disappear.
On the other hand, for the boundary terms, we will have 
\begin{equation} \label{abstract-boundary-dis} 
\mathscr{L}^{\text{dual}}_{0,1}
=     - e^{\left( \bar{K}^{0}_{0} +s \right)^{-1}\bar{K}^{+}_{0} \bar{K}^{-}_{1}} \left( \psi( \bar{K}^{0}_{1} +s) - \psi(2s)  \right)   e^{-\left( \bar{K}^{0}_{0} +s \right)^{-1}\bar{K}^{+}_{0} \bar{K}^{-}_{1}} \,,
 \end{equation}  
and 
 \begin{equation} \label{abstract-boundaryN-dis} 
\mathscr{L}^{\text{dual}}_{N,N+1}
=
-  e^{\bar{K}^{-}_{N}\left( \bar{K}^{0}_{N+1} +s \right)^{-1}\bar{K}^{+}_{N+1} } \left( \psi( \bar{K}^{0}_{N} +s) - \psi(2s)  \right)   e^{-\bar{K}^{-}_{N}\left( \bar{K}^{0}_{N+1} +s \right)^{-1}\bar{K}^{+}_{N+1} } \,.
 \end{equation}  

\begin{remark}
Consider the representation of the $\mathfrak{sl}(2)$
algebra given by
$$
K^{+}g  (n)= (2s+n-1)g(n-1)\,,\qquad 
 K^{-}g (n)=(n+1)g(n+1)\,,\qquad 
 K^{0}g (n)=(n+s)g(n) \,,
$$
which satisfies the commutation relations
$$
 [{K}^{0}, {K}^{\pm}]=\pm {K}^{\pm} \qquad  \text{and} \qquad [{K}^{+},{K}^{-}]=-2{K}^{0}\;.
$$
This representation is just the transpose of the 
one introduced in \eqref{discretesu11}, i.e.
\be
\label{transpose}
(\bar{K}^a)^t=K^a \qquad\qquad a\in\{+,-,0\}\,.
\ee
{The local density of the spin chain with Heisenberg XXX Hamiltonian \cite[(4.17)]{frassek2021exact} is then recovered from the negative transpose of the {local} Markov generator  $\mathscr{L}^{dual}_{i,i+1}$ after using the identity }
\begin{equation}
 \begin{split}
  e^{-\bar{K}^{-}_{i}(\bar{K}^{0}_{i}+s)^{-1}\bar{K}^{+}_{i+1}} (\psi(\bar{K}^{0}_{i+1} +s)-\psi(2s))
  e^{\bar{K}^{-}_{i}(\bar{K}^{0}_{i}+s)^{-1}\bar{K}^{+}_{i+1}}
 - h_s(\bar{K}^0_i-s) +h_s(\bar{K}^0_{i+1}-s)\\
 =  e^{\bar{K}^{-}_{i}(\bar{K}^{0}_{i+1}+s)^{-1}\bar{K}^{+}_{i+1}} (\psi(\bar{K}^{0}_{i} +s)-\psi(2s))
e^{-\bar{K}^{-}_{i}(\bar{K}^{0}_{i+1}+s)^{-1}\bar{K}^{+}_{i+1}}   
 \,.
\end{split}
 \end{equation} 
 The identity above follows from the equivalence of the action of the bulk generator \eqref{abstract-right-bulk-dis} for $i=N$  and  the boundary generator \eqref{abstract-boundaryN-dis} 
(that is indeed verified by direct computation in \eqref{eq:bbulk2} and \eqref{eq:hopB} respectively).
\end{remark}

\subsubsection{Bulk}
We now compute the action of the generator \eqref{abstract-right-bulk-dis} on functions $g(n)$. The action the generator \eqref{abstract-left-bulk-dis} then follows.  For this purpose we consider the following object
\begin{equation} \label{abstract-dis2} 
   O_+(\gamma)= -e^{-\gamma \bar{K}^{+}} (\psi(\bar{K}^{0} +s)-\psi(2s))
e^{ \gamma \bar{K}^{+}}\,.
\end{equation}
When computed  on site $i$ and for $\gamma = \bar{K}^{-}_{i+1}(\bar{K}^{0}_{i+1}+s)^{-1} $ it matches the right hand side of equation \eqref{abstract-left-bulk-dis},
up to the diagonal terms. 

We proceed now by an explicit computation of the action of \eqref{abstract-dis2} on function $g:\mathbb{N}_{0}\to\R$ {with compact support}. The inner part gives that
\begin{equation*}
    e^{ \gamma \bar{K}^{+}} g(n) = \sum_{j=0}^{\infty} \dfrac{\gamma^{j}}{j!}\left( \bar{K}^{+} \right)^{j} g(n) = \sum_{j=0}^{\infty} \dfrac{\gamma^{j}}{j!} \dfrac{\Gamma(n+2s+j)}{\Gamma(n+2s)}  g(n+j) \;.
\end{equation*}
Furthermore, since $\bar{K}^0$ acts diagonally, recalling \eqref{digammadifference} we have
\begin{equation} \label{identityko}
\left(\psi(\bar{K}^{0} +s)-\psi(2s)\right) g(n) = h_{s}(n) g(n)\,.
\end{equation}
Thus we get  
\begin{equation}
    e^{-\gamma \bar{K}^{+}} (\psi(\bar{K}^{0} +s)-\psi(2s))
e^{ \gamma \bar{K}^{+}} g(n) = \sum_{l=0}^{\infty} \sum_{j=0}^{\infty} (-1)^{l} \dfrac{\gamma^{j+l}}{j! l!} \dfrac{\Gamma(n+2s+j+l)}{\Gamma(n+2s)}  h_{s}(n+l) g(n+j+l)\,.
\end{equation}
Performing a change of variables and then changing the order of summation we have 
\begin{equation}
    e^{-\gamma \bar{K}^{+}} (\psi(\bar{K}^{0} +s)-\psi(2s))
e^{ \gamma \bar{K}^{+}} g(n) = \sum_{k=0}^{\infty} \gamma^{k}\sum_{l=0}^{k} \dfrac{(-1)^{l} }{ l! (k-l)!} \dfrac{\Gamma(n+2s+k)}{\Gamma(n+2s)}  h_{s}(n+l) g(n+k)\,.
\end{equation}
Separating the term for $k=0$ and using the identity \eqref{eq:sumform1} for $k>0$, i.e.
\begin{equation}
\sum_{l=0}^{k} \dfrac{(-1)^{l} }{ l! (k-l)!} \dfrac{\Gamma(n+2s+k)}{\Gamma(n+2s)}  h_{s}(n+l)= - \dfrac{1}{k} \;,
\end{equation}
we find that 
\begin{equation}
O_+(\gamma)g(n) 
= -h_s(n) g(n) + \sum_{k=1}^{\infty} \dfrac{\gamma^{k}}{k} g(n+k) \;.
\end{equation}
Now we consider the operator $O_+(\gamma)$ above on site $i$, we set $\gamma = \bar{K}^{-}_{i+1}(\bar{K}^{0}_{i+1}+s)^{-1} $ and we act on a function $f : \N^{N+2}\to\R$ to get 
\begin{eqnarray}
- e^{-\bar{K}^{+}_{i}\bar{K}^{-}_{i+1}(\bar{K}^{0}_{i+1}+s)^{-1}} (\psi(\bar{K}^{0}_{i} +s)-\psi(2s))
e^{ \bar{K}^{+}_{i}\bar{K}^{-}_{i+1}(\bar{K}^{0}_{i+1}+s)^{-1}} f(\xi) = \\ \sum_{k=1}^{\infty} \dfrac{ \left(\bar{K}^{-}_{i+1} (\bar{K}^{0}_{i+1}+s)^{-1} \right)^{k}}{k} f(\xi+k\delta_i) - h_s(\xi_i)f(\xi) =  \\ \sum_{k=1}^{\xi_{i+1}} \dfrac{ \xi_{i+1}!}{k (\xi_{i+1} -k)!} \dfrac{\Gamma(\xi_{i+1}+2s+k)}{\Gamma(\xi_{i+1}+2s)} f(\xi + k \delta_i - k \delta_{i+1}) - h_s(\xi_i)f(\xi) = \\
\sum_{k=1}^{\xi_{i+1}} \varphi_s(k, \xi_{i+1})  f(\xi + k \delta_i - k \delta_{i+1}) - h_s(\xi_i)f(\xi) = \\
\left( {\mathscr{L}^{\leftarrow,\text{dual}}_{i,i+1}} + h_s(\xi_{i+1})- h_s(\xi_i) \right) f(\xi)\,.
\end{eqnarray}
Thus \eqref{abstract-left-bulk-dis} is proved. Similarly,
\begin{eqnarray}\label{eq:bbulk2}
- e^{-\bar{K}^{-}_{i}(\bar{K}^{0}_{i}+s)^{-1}\bar{K}^{+}_{i+1}} (\psi(\bar{K}^{0}_{i+1} +s)-\psi(2s))
e^{ \bar{K}^{-}_{i}(\bar{K}^{0}_{i}+s)^{-1}\bar{K}^{+}_{i+1}} f(\xi)
=
\\ \left( {\mathscr{L}^{ \rightarrow,\text{dual}}_{i,i+1}} + h_s(\xi_{i})- h_s(\xi_{i+1}) \right) f(\xi)\,,
\end{eqnarray}
from which \eqref{abstract-right-bulk-dis} is established.

\subsubsection{Boundary}
As before, consider the following object
\begin{equation} \label{abstract-dis1} 
   O_-(\gamma)= -e^{\gamma \bar{K}^{-}} (\psi(\bar{K}^{0} +s)-\psi(2s))
e^{ -\gamma \bar{K}^{-}}\,.
\end{equation}
When computed  on site $1$ and for $\gamma = (\bar{K}^{0}_{0}+s)^{-1} \bar{K}^{+}_{0}$ it matches the right hand side of equation \eqref{abstract-boundary-dis}.
We start by computing the action of the inner element:
\begin{equation*}
    e^{- \gamma \bar{K}^{-}} g(n) = \sum_{j=0}^{\infty} \dfrac{(-\gamma)^{j}}{j!} \left( \bar{K}^{-}\right)^{j}g(n) = \sum_{j=0}^{n} (-1)^j \dfrac{\gamma^{j}}{j!} \dfrac{n!}{(n-j)!} g(n-j) \;.
\end{equation*}
Next, using again the identity \eqref{identityko}, we have
\begin{equation*}
  (\psi(\bar{K}^{0} +s)-\psi(2s))  e^{- \gamma \bar{K}^{-}} g(n) = 
h_s(n)\sum_{j=0}^{n} (-1)^j \dfrac{\gamma^{j}}{j!} \dfrac{n!}{(n-j)!} g(n-j)\,.
\end{equation*}
Thus we arrive to
\begin{equation*}
   e^{\gamma \bar{K}^{-}} (\psi(\bar{K}^{0} +s)-\psi(2s))
e^{ -\gamma \bar{K}^{-}} g(n) 
   = \sum_{j=0}^{n}  \sum_{l=0}^{n-j}  (-1)^j  \dfrac{\gamma^{l+j}}{l!j!}\dfrac{n!}{(n-l-j)!} h_{s}(n-l)
 g(n-l-j) 
\end{equation*}
Performing the change of variable $k=l+j$ and changing the order of summation  we find that the above corresponds to
\begin{equation*}
\sum_{k=0}^{n} \gamma^{k} g(n-k)  \sum_{l=0}^{k}  (-1)^{k-l} \binom{n}{l} \binom{n-l}{k-l}  h_{s}(n-l)
 \;.
\end{equation*}
Separating the case $k=0$ together with the identity   cf.~\cite[(3.13)]{frassek2021exact}
\begin{equation*}
\sum_{l=0}^{k}  (-1)^{k-l} \binom{n}{l} \binom{n-l}{k-l}  h_{s}(n-l) = - \varphi_s(k, n)  \,,
\end{equation*}
we find that 
\begin{equation}
     O_-(\gamma)= \sum_{k=1}^{n} \gamma^{k} \varphi_s(k, n) g(n-k) - h_{s}(n)g(n)
\end{equation}
Setting $\gamma= (\bar{K}^{0}_{0} + s)^{-1}\bar{K}^{+}_{0}$ and evaluating the above expression in site $i=1$ on functions $f : \N^{N+2}\to\R$ we find 
\begin{eqnarray}\label{eq:hopB}
 &&- e^{(\bar{K}^{0}_{0} + s)^{-1}\bar{K}^{+}_{0} \bar{K}_1^{-}} \left( \psi( \bar{K}^{0}_1 +s) - \psi(2s)  \right)   e^{-(\bar{K}^{0}_{0} + s)^{-1}\bar{K}^{+}_{0} \bar{K}_1^{-}} f(\xi)  \nonumber\\
&&\qquad=\sum_{k=1}^{\xi_1} \varphi_s(k, \xi_1) f(\xi -k \delta_1 + k \delta_0 ) - h_{s}(\xi_1)f(\xi) =\mathscr{L}^{dual}_{0,1}f(\xi)\;.
\end{eqnarray}

\subsection{Duality and change of representation}

In Section~\ref{sec3.1} we proved the duality relation between the heat conduction model and the absorbing particle process through a direct computation. Here we argue that duality can elegantly be obtained from the algebraic descriptions of the previous sections. 
The idea behind the algebraic approach is that one can deduce  duality relations for  Markov processes as a consequence of a change of representation of the underlying algebra. For the case of interest, the argument is described in the following.

First we note that the two representations \eqref{ctssu11} and  \eqref{discretesu11} satisfy the
duality relation  
\be
\label{dual-basic1}
\mathscr{K}^a d(\cdot,n)(z) = \bar K^a d(z,\cdot)(n) \qquad\qquad a\in \{+,-,0\}
\ee
for the single-site duality function 
$d:\mathbb{R}^{+} \times \mathbb{N}_{0} \rightarrow \mathbb{R}$ given by
\be
d(z,n) =z^{n} \dfrac{\Gamma(2s)}{\Gamma(2s+n)} \;.
\ee 
Furthermore, the two representations \eqref{ctssu11-site0} and  \eqref{discretesu11} satisfy the
duality relation
\be
\label{dual-basic2}
\mathscr{S}^a b(\cdot,n)(z) = \bar K^a b(z,\cdot)(n) \qquad\qquad a\in \{+,-,0\}
\ee
for the single-site duality function 
\be
b(z,n) =z^{n} \;.
\ee 

After having expressed the Markov processes algebraically, the duality relations of the Lie algebra generators above can be used to infer a duality relation between the  processes when regarding the Markov generator as an expansion in terms of elements of the universal enveloping algebra of $\mathfrak{sl}(2)$. The key idea is contained in the following general theorem
(for a proof see Theorem 2.1 of \cite{carinci2015dualities}, see also \cite{franceschini2018self} in the context of self-duality):
\begin{theorem}
Let $\mathfrak{g}$ be a Lie algebra generated by $\left \{ A_i \right \}_{i=1}^{n}$ and let $\left \{ B_i \right \}_{i=1}^{n}$ be generators of the conjugate Lie algebra.
    If, for every $i=1 , \ldots, n$, $A_i$ is dual to $B_i$ with duality function $D$ then the element $(A_{i_1})^{n_1} \dots (A_{i_k})^{n_k} $ is dual with the same duality function $D$ to the element $(B_{i_k})^{n_k} \dots (B_{i_1})^{n_1} $ for all $k \in \N $ and for all $n_1, \ldots, n_k \in \mathbb{N}$.
\end{theorem}
The rule of thumb of the theorem is that a given sequence of $A$ operators is dual to a sequence of $B$ operators but in reversed order. 
Combining together the above theorem
with the algebraic expression of the Markov generators  found in Section~\ref{algebraic-original} and in Section~\ref{algebraic-dual} one deduces that the algebra dualities \eqref{dual-basic1} and
\eqref{dual-basic2} imply
the duality of Theorem~\ref{duality-thm} with
\be
\label{duality-abstract}
D(y,\xi) = b_0(T_L,\xi_0) \Big[ \prod_{i=1}^N d(y_i,\xi_i)\Big] b(T_R,\xi_{N+1})\,,
\ee
for $y\in\R_+^N$ and $\xi\in \N_0^{N+2}$.

\begin{remark}
The duality function \eqref{duality-abstract} thus emerges
from 
the representation theory of $\mathfrak{sl}(2)$ and is indeed
the duality function of several
other models with the same underlying
Lie algebra \cite{carinci2015dualities,FGS}.
Via the algebraic approach one can
obtain other duality functions as well,
e.g. classical orthogonal polynomials \cite{carinci2019orthogonal, floreani2022orthogonal}, or extend the dualities to asymmetric systems, 
see for instance \cite{CFG, franceschini2022orthogonal} and references therein.
For the asymmetric version of the harmonic model studied in this paper  we refer the reader to \cite{Sasamoto,Povolotsky,barraquand,Frassek:2022fjs}.
\end{remark}

{
\paragraph{Acknowledgment.} This work has been performed under the  auspices of the ``Istituto Nazionale di Alta Matematica - F.Severi''
(INDAM). }

\appendix
\section{Reversibility at equilibrium} \label{appendixa}
{
Here we show that, in equilibrium setting $\lambda_L = \lambda_R = \lambda$, the generator $\mathscr{L}$ in \eqref{eq:gen} satisfies $\langle \mathscr{L}f,g\rangle = \langle f, \mathscr{L} g\rangle$ for all polynomial functions $f,g$. 
The scalar product above is
$$\langle f,g\rangle = \int_{\R_+^N} f(y) g(y) \mu_N(dy)\,,
$$ 
with $\mu_N$  the probability measure in
\eqref{eq:density}. 
In order to show this, we prove that $\langle \mathscr{L}_{i,i+1}f,g\rangle = \langle f, \mathscr{L}_{i,i+1}g\rangle$  and
similarly $\langle \mathscr{L}_{1}f,g\rangle = \langle f, \mathscr{L}_{1}g\rangle$ 
and  $\langle \mathscr{L}_{N}f,g\rangle = \langle f, \mathscr{L}_{N}g\rangle$.}
Using \eqref{eq:gen}
the claim then follows.

To alleviate notation, when considering the action of $\mathscr{L}_{i,i+1}$ on a function $g$ we write only the variables $y_i$ and $y_{i+1}$ and skip the
dependence of the functions 
on all the other variables,
which stay untouched. Similarly we do not write all the integrals
involving variables other than $y_i$ and $y_{i+1}$.
Thus, recalling the decomposition
$\mathscr{L}_{i,i+1} = \mathscr{L}^{\rightarrow}_{i,i+1} +\mathscr{L}^{\leftarrow}_{i,i+1}$ and 
modulo the abuse of notation, we may write
\begin{eqnarray*}
&& \left( \dfrac{\Gamma(2s)}{\lambda^{2s}}\right)^{2} \cdot\langle  \mathscr{L}_{i,i+1} f, g \rangle = \\
&& 
 \int_{\R^+} dy_i \int_{\R_+} dy_{i+1}  \int_{0}^{y_i}  \frac{d\alpha}{\alpha} \left(1-\frac{\alpha}{y_i}\right)^{2s-1} f(y_i-\alpha,y_{i+1}+  \alpha)
 g(y_i, y_{i+1} )  
(y_{i} y_{i+1})^{2s-1} e^{-\lambda(y_{i} + y_{i+1})} \\
 &&  
 \qquad
 - \int_{\R_+} dy_i \int_{\R_+} dy_{i+1}  \int_{0}^{y_i}  \frac{d\alpha}{\alpha} \left(1-\frac{\alpha}{y_i}\right)^{2s-1}f(y_i, y_{i+1})
 g(y_i, y_{i+1} )  (y_{i}y_{i+1})^{2s-1} e^{-\lambda(y_{i} + y_{i+1})}\\
&& 
+ \int_{\R_+} dy_i \int_{\R_+} dy_{i+1}  \int_{0}^{y_{i+1}}  \frac{d\alpha}{\alpha} \left(1-\frac{\alpha}{y_{i+1}}\right)^{2s-1} f(y_i+\alpha,y_{i+1}-  \alpha)
 g(y_i, y_{i+1} )  
   (y_{i} y_{i+1})^{2s-1} e^{-\lambda(y_{i} + y_{i+1})} \\
 &&  
 \qquad
 - \int_{\R_+} dy_i \int_{\R_+} dy_{i+1}  \int_{0}^{y_{i+1}}  \frac{d\alpha}{\alpha} \left(1-\frac{\alpha}{y_{i+1}}\right)^{2s-1}f(y_i, y_{i+1})
 g(y_i, y_{i+1} )   (y_{i}y_{i+1})^{2s-1} e^{-\lambda(y_{i} + y_{i+1})} \,.
\end{eqnarray*}
We proceed formally and consider the first multiple integral on the right hand side of the above display. 
Applying (twice) Fubini's theorem and changing to new variables $z_i = y_i - \alpha$ and $z_{i+1} =y_{i+1} + \alpha $
this can be rewritten as
$$
\int_{\R_+} dz_i \int_{\R_+} dz_{i+1}  \int_{0}^{z_{i+1}}  \frac{d\alpha}{\alpha} \left(1-\frac{\alpha}{z_{i+1}}\right)^{2s-1} f(z_i, z_{i+1})
 g(z_i + \alpha, z_{i+1} - \alpha )   (z_{i}z_{i+1})^{2s-1} e^{-\lambda(z_{i} + z_{i+1})}\,.
$$
Similarly, applying Fubini's theorem and using  the change of variables $z_i = y_i + \alpha$ and $z_{i+1} =y_{i+1} - \alpha$, the third multiple integral on the right hand side can be rewritten as
$$
\int_{\R_+} dz_i \int_{\R_+} dz_{i+1}  \int_{0}^{z_{i}}  \frac{d\alpha}{\alpha} \left(1-\frac{\alpha}{z_{i}}\right)^{2s-1} f(z_i, z_{i+1})
 g(z_i - \alpha, z_{i+1} + \alpha )  (z_{i}z_{i+1})^{2s-1} e^{-\lambda(z_{i} + z_{i+1})}\,.
$$
All in all, one gets
\begin{eqnarray*}
&&\left( \dfrac{\Gamma(2s)}{\lambda^{2s}}\right)^{2} \cdot\langle  \mathscr{L}_{i,i+1} h, g \rangle = \\
&& 
\int_{\R_+} dz_i \int_{\R_+} dz_{i+1}  \int_{0}^{z_{i+1}}  \frac{d\alpha}{\alpha} \left(1-\frac{\alpha}{z_{i+1}}\right)^{2s-1} f(z_i, z_{i+1})
 g(z_i + \alpha, z_{i+1} - \alpha )  (z_{i}z_{i+1})^{2s-1} e^{-\lambda(z_{i} + z_{i+1})} \\
 && 
 \qquad
 - \int_{\R_+} dy_i \int_{\R_+} dy_{i+1}  \int_{0}^{y_i}  \frac{d\alpha}{\alpha} \left(1-\frac{\alpha}{y_i}\right)^{2s-1}f(y_i, y_{i+1})
 g(y_i, y_{i+1} )   (y_{i}y_{i+1})^{2s-1} e^{-\lambda(y_{i} + y_{i+1})}\\
 && 
 +\int_{\R^+} dz_i \int_{\R_+} dz_{i+1}  \int_{0}^{z_i}  \frac{d\alpha}{\alpha} \left(1-\frac{\alpha}{z_i}\right)^{2s-1} f(z_i, z_{i+1})g(z_i-\alpha,z_{i+1}+  \alpha)
  (z_{i} z_{i+1})^{2s-1} e^{-\lambda(z_{i} + z_{i+1})} \\
 && 
 \qquad
 - \int_{\R_+} dy_i \int_{\R_+} dy_{i+1}  \int_{0}^{y_{i+1}}  \frac{d\alpha}{\alpha} \left(1-\frac{\alpha}{y_{i+1}}\right)^{2s-1}f(y_i, y_{i+1})
 g(y_i, y_{i+1} )   (y_{i}y_{i+1})^{2s-1} e^{-\lambda(y_{i} + y_{i+1})}\,.
\end{eqnarray*}
Combining together the first and the fourth lines of the r.h.s. of the above display, we recognize the action of $\mathscr{L}^{\leftarrow}_{i,i+1}$ on the function $g$. 
Similarly, combining together the second and the third lines, we recognize the action of $\mathscr{L}^{\rightarrow}_{i,i+1}$ on the function $g$. 
So we have proved that 
$\langle  \mathscr{L}_{i,i+1} f, g \rangle = \langle  f,  \mathscr{L}_{i,i+1} g \rangle$. 

For the boundary terms a similar computation holds. Let's consider the generator of the left reservoirs as the right one is similar. We have
\begin{align*}
\langle  \mathscr{L}_{1} h, g \rangle & = \int_{0}^{+ \infty} dy_{1} \int_{0}^{y_1} \dfrac{d\alpha}{\alpha} \left( 1-\dfrac{\alpha}{y_1} \right)^{2s-1}  h(y_1 - \alpha) g(y_1) \dfrac{\lambda^{2s}}{\Gamma(2s)} y_1^{2s-1} e^{- \lambda y_1} \\
& - \int_{0}^{+ \infty} dy_{1} \int_{0}^{y_1} \dfrac{d\alpha}{\alpha} \left( 1-\dfrac{\alpha}{y_1} \right)^{2s-1}  h(y_1 ) g(y_1) 
\dfrac{\lambda^{2s}}{\Gamma(2s)} y_1^{2s-1} e^{- \lambda y_1} \\
& + \int_{0}^{+ \infty} dy_{1} \int_{0}^{+ \infty} \dfrac{d\alpha}{\alpha} e^{-\lambda \alpha} h(y_1 + \alpha)g(y_1) 
\dfrac{\lambda^{2s}}{\Gamma(2s)} y_1^{2s-1} e^{- \lambda y_1} \\
& -  \int_{0}^{+ \infty} dy_{1} \int_{0}^{+ \infty} \dfrac{d\alpha}{\alpha} e^{-\lambda \alpha} h(y_1)g(y_1) \dfrac{\lambda^{2s}}{\Gamma(2s)} y_1^{2s-1} e^{- \lambda y_1}\,.
\end{align*}
As before, let's start by considering the first double integral of the right hand side. By Fubini and the change of variable $y_1 = z_1 + \alpha$ we find
\begin{eqnarray*}
&& \int_{0}^{+ \infty} dy_1
 \int_{0}^{y_1} \dfrac{d\alpha }{\alpha}
 \left( 1 - \frac{\alpha}{y_1} \right)^{2s-1}  h(y_1 - \alpha) g(y_1) \frac{\lambda^{2s}}{\Gamma(2s)}y_1^{2s-1} e^{- \lambda y_1}\\
 &&= 
  \int_{0}^{+ \infty} dz_1
  \int_{0}^{+ \infty} \frac{d\alpha}{\alpha}e^{- \lambda \alpha}  h(z_1) g(z_1+ \alpha) \frac{\lambda^{2s}}{\Gamma(2s)} z_{1}^{2s-1} e^{- \lambda z_1}\,.
\end{eqnarray*}
Similarly, applying Fubini and changing variable $y_1 = z_1 - \alpha$  to the third double integral we get
\begin{eqnarray*}
&&\int_{0}^{+ \infty} dy_1
\int_{0}^{+ \infty} \frac{d\alpha}{\alpha}e^{- \lambda \alpha} h(y_1+\alpha) g(y_1) \frac{\lambda^{2s}}{\Gamma(2s)} y_1^{2s-1} e^{- \lambda y_1} \\
&&= 
\int_{0}^{+ \infty} dz_{1} \int_{0}^{z_1} \dfrac{d\alpha}{\alpha} \left( 1-\dfrac{\alpha}{z_1} \right)^{2s-1}  h(z_1) g(z_1 - \alpha) \frac{\lambda^{2s}}{\Gamma(2s)} z_1^{2s-1} e^{- \lambda z_1}\,.
\end{eqnarray*}
Substituting above, this shows that $\langle  \mathscr{L}_1 h, g \rangle = \langle  h,  \mathscr{L}_1 g \rangle$.

\section{Shift formulas}
\label{B}
In this appendix we collect two useful formulas. First we note that  the Taylor series of the function $f(x+\alpha)$ around $\alpha=0$ can be written as
\begin{equation}
\label{B1}
f(x+\alpha)=\sum_{k=0}^\infty\frac{\alpha^k}{k!}f^{(k)}(x)= e^{\alpha \partial_x}f(x)\,.
\end{equation} 
Further, the Taylor series of the function $f(\alpha x)$ around $\alpha=1$ gives
\begin{equation}
\begin{split}
\label{B2}
f(\alpha x)&=\sum_{k=0}^\infty\frac{(\alpha-1)^k}{k!}x^kf^{(k)}(x)\\&=\sum_{k=0}^\infty\frac{(\alpha-1)^k}{k!}\frac{\Gamma(x\partial_x+1)}{\Gamma(x\partial_x+1-k)} f^{}(x)\\
&= \alpha^{x \partial_x}f(x)\,,
\end{split}
\end{equation} 
where in the first equality it has been used that
\begin{equation}
 x^k\partial_x^k=\frac{\Gamma(x\partial_x+1)}{\Gamma(x\partial_x+1-k)}\,,
\end{equation} 
{
while the second equality employes again Taylor's expansion.
}

{
\small
\bibliographystyle{utphys2}
\bibliography{refs}
}

\end{document}